\documentclass[journal,onecolumn,draftclsnofoot,12pt]{IEEEtran}

\IEEEoverridecommandlockouts

\makeatletter
\def\markboth#1#2{\def\leftmark{\@IEEEcompsoconly{\sffamily}\MakeUppercase{\protect#1}}%
\def\rightmark{\@IEEEcompsoconly{\sffamily}\MakeUppercase{\protect#2}}}
\makeatother

\usepackage[utf8x]{inputenc}
\usepackage[english]{babel}
\usepackage{ucs}
\usepackage{amsmath}
\usepackage{amsfonts}
\usepackage{amssymb}
\usepackage{amsthm}
\usepackage{array}
\usepackage{verbatim}
\usepackage{listings}
\usepackage{stfloats}

\usepackage{algorithm}
\usepackage{algorithmic}
\usepackage{url}
  \usepackage{enumerate}
  \usepackage{cite}
  \usepackage{color}
 
\ifCLASSINFOpdf
  \usepackage[pdftex]{graphicx}
  \usepackage{subfigure}
\usepackage{epstopdf}
\graphicspath{{./}} % for arxiv
   \DeclareGraphicsExtensions{.eps,.ps,.png}
\else
  \usepackage[dvipdf]{graphicx}
  \graphicspath{{./img/}}
  \usepackage{subfigure}
   \DeclareGraphicsExtensions{.eps,.ps}   
\fi

\newcommand{\Hb}{\mathbf{H}}
\newcommand{\Sb}{\mathbf{\Sigma}}

\newcommand{\F}{\mathbf{F}}

\newcommand{\I}{\mathbf{I}}
\newcommand{\Y}{\mathbf{Y}}

\newcommand{\Z}{\mathbf{Z}}

\newcommand{\X}{\mathbf{X}}

\newcommand{\zero}{\mathbf{0}}

\newcommand{\x}{\mathbf{x}}

\newcommand{\y}{\mathbf{y}}

\newcommand{\vmu}{\mathbf{\boldsymbol{\mu}}}

\newcommand{\z}{\mathbf{z}}

\newcommand{\tr}{\textnormal{tr}}
\newcommand{\Ex}[2]{{\textnormal{E}_{#1}\left[#2\right]}}
\newcommand{\CInf}[3]{{\textnormal{I}\left(#1;#2|#3\right)}}
\newcommand{\Inf}[2]{{\textnormal{I}\left(#1;#2\right)}}

\newcommand{\Ent}[1]{{\textnormal{H}\left(#1\right)}}

\newcommand{\dEnt}[1]{{\textnormal{h}\left(#1\right)}}

\newtheorem{lemma}{Lemma}

\newtheorem{remark}{Remark}
\newtheorem{corollary}{Corollary}

\newtheorem{proposition}{Proposition}

\title{Unified Capacity Limit of Non-Coherent Wideband Fading Channels}

\author{Felipe G\'omez-Cuba$^*$, \emph{Student Member, IEEE}, Jinfeng Du$^*$, \emph{Member, IEEE}, Muriel M{\'e}dard, \emph{Fellow,
IEEE},  and Elza Erkip, \emph{Fellow, IEEE}
\thanks{This work was presented in part at IEEE International Symposium on Information Theory, Hong Kong, July 2015~\cite{CDME15}.}
\thanks{This work is supported in part by the MIT Wireless Center, the National Science Foundation under Grant 1302336, the NYU WIRELESS, and the Swedish Research Council (VR) under Grant 637-2013-473. This work is also funded by FPU12/01319 MECD, University of Vigo and Xunta de Galicia, Spain.}
\thanks{F. G\'omez-Cuba is with AtlantTIC, University of Vigo, C.P. 36310 Vigo, Spain  (Email:  fgomez@gti.uvigo.es).}%
\thanks{J. Du is with Nokia Bell Labs, Holmdel, NJ 07733, USA. He was with Massachusetts Institute of Technology, Cambridge, MA, USA  (Email: jinfeng.du@bell-labs.com). 
}
\thanks{M. M\'edard is with  Massachusetts Institute of Technology,
Cambridge, MA 02139, USA (Email: medard@mit.edu).}
\thanks{E. Erkip is with NYU Tandon School of Engineering, Brooklyn, NY 11201, USA (Email: elza@nyu.edu).}%
\thanks{$^*$The first two authors contributed equally to this work.}
}
\begin{document}
\maketitle 

\begin{abstract}
In non-coherent wideband fading channels where energy rather than spectrum is the limiting resource, peaky and non-peaky signaling schemes have long been considered species apart, as the first approaches asymptotically the capacity of a wideband AWGN channel with the same average SNR, whereas the second reaches a peak rate at some finite \emph{critical bandwidth} and then falls to zero as bandwidth grows to infinity. In this paper it is shown that this  distinction is in fact an artifact of the limited attention paid in the past to the product between the bandwidth and the fraction of time it is in use. This fundamental quantity, called \emph{bandwidth occupancy}, measures average bandwidth usage over time.  
For all signaling schemes with the same bandwidth occupancy, achievable rates approach to the wideband AWGN capacity within the same gap  as the bandwidth occupancy approaches its critical value, and decrease to zero as the occupancy goes to infinity. This unified analysis produces quantitative closed-form expressions for the ideal bandwidth occupancy, recovers the existing capacity results for (non-)peaky signaling schemes, and unveils a trade-off between the accuracy of approximating capacity with a generalized Taylor polynomial and the accuracy with which the optimal bandwidth occupancy can be bounded.

\end{abstract}


\section{Introduction and Related Work}
\label{sec:introduction}
Recently there has been great interest in wireless channels with a large bandwidth, owing in part to the prospective investments onto the millimeter wave bands, where vast quantities of new spectrum is readily available  \cite{Pi2011,PietBRPC:12,RanRapEr:14,Rappaport2014-mmwbook}.
In a frequency selective fading channel where there is no channel state information at the receiver (CSIR) or the transmitter, the wideband capacity is affected by the growing uncertainty in the channel impulse response. As bandwidth grows while energy is constrained, it becomes infeasible to estimate the channel coefficients to a precision sufficient for coherent detection. Moreover, if the transmitted signal power {}{is spread} across all the available bandwidth and time slots, the desired signal would be buried by the channel uncertainty {}{when} bandwidth is too large. 
M\'edard and Gallager proved this \cite{journals/tit/MedardG02} through an upper bound to the rate that is proportional to the ratio between the fourth moment of the signal ($\Ex{}{|x|^4}$) and its bandwidth ($B$), i.e., $R<\propto {\Ex{}{|x|^4}}/{B}$. That is, to achieve rates above zero when $B\to\infty$, one has to make $\Ex{}{|x|^4}$ grow at least as fast as $B$ by concentrating the power of the signal in a vanishing fraction of its transmitted symbols (i.e. infrequent bursts of very large power).

In this paper we investigate the capacity bounds of non-coherent wideband fading channels in multi-input multi-output (MIMO)  setup where both the signaling bandwidth and signal peakiness are design parameters.
The channel is assumed to be rich scattering, frequency selective, block fading with a coherence time $T_c$ and a delay spread $D$, such that the channel frequency response becomes uncorrelated for frequencies apart from more than one coherence bandwidth $B_c{\triangleq} 1/D$. The channel coherence length, $B_cT_c$, is assumed to be large for capacity analysis purposes{}{, as in almost all practical channels, $B_cT_c\gg1$}. In our expressions we temporarily treat $B_cT_c$ as a fixed parameter to derive closed-form expressions, where approximation errors originated from $B_cT_c{\gg} 1$ are highlighted in small-$o$ expressions parametrized by higher order terms of $B_cT_c$. We further assume that $B_cT_c{>}N_t$, which is easily satisfied in typical systems where the number of transmit antennas is not massive. 
We generalize the analysis method in \cite{journals/twc/LozanoP12}, developed for non-peaky signaling in single-input single-output (SISO)  systems, to MIMO systems and extend it to arbitrary level of signal peakiness by enforcing a transmission duty cycle $\delta{\in}(0,1]$. The duty cycle prescribes a bursty transmission scheme where the transmitter is active only for a fraction $\delta$ of time with boosted signal power $P/\delta$ harnessed from the $(1{-}\delta)$ silent-cycle. 
  {Denoting by $C(B)$ the capacity of the unconstrained non-coherent channel} and by $C(B,\delta)$ the maximal rate achieved by using bandwidth $B$ and duty cycle $\delta${}{, for all $B>0$ and  $\delta\in(0,1]$, we have} 
\[ C(B,\delta)< C(B) \leq C^\infty\triangleq {N_{\mathrm{r}}P}/{N_0} \mbox{ [nats/s]},\] 
where $C^{\infty}$ is the limit capacity of the coherent channel at infinite bandwidth, $P$ is the received signal power, $N_0$ is the noise power spectral density, and $N_{\mathrm{r}}$ is the number of receive antennas.
{}{Note that} the first inequality is strict because we do not exploit the position of the active symbols to convey information.
 We show in Sec.~\ref{sec:Bdlimit} that  $C(B, \delta)$ is upper and lower bounded by 
\[ R^{\mathrm{LB}}(\delta B) \leq C(B, \delta) \leq R^{\mathrm{UB}}(\delta B).\]
Note that both the upper and lower bounds, up to a small approximation error $o(1/\delta B)$, depend on $B$ and $\delta$ only through the product $\delta B$, which measures average bandwidth usage over time and is named the ``bandwidth occupancy''.
 Our results show that for a series of signaling schemes with finite signaling bandwidth $B$ larger or equal to a \textit{critical bandwidth occupancy} $(\delta B)_{\mathrm{crit}}$,  which falls in a range prescribed by closed-form expressions,
it is possible to achieve rates close to  $C^\infty$ within the same rate penalty
 \begin{align}\label{eqn:low:peak}
 C(B, \delta) \geq N_\mathrm{r}\frac{P}{N_0}-{\Delta_C},
\quad {\Delta_C}=N_\mathrm{r}\frac{P}{N_0}\sqrt{\frac{1+\log B_cT_c}{B_cT_c}(\kappa-2+N_{\mathrm{t}}+N_{\mathrm{r}})\log\pi},
 \end{align} 
as long as the duty cycle is $\delta{=}\frac{(\delta B)_{\mathrm{crit}}}{B}$. Here  $N_{\mathrm{t}}$ is the number of transmit antennas and  $\kappa{>}0$ is the kurtosis (whose definition is deferred to Sec.~\ref{sec:model}) of the channel. {}{Thus}, it is possible to approach $C^\infty$ up to the same gap with {}{any} $\delta{\in}(0,1]$. {}{Note also that} $B_cT_c{\gg}1$ leads to $\Delta_C {\simeq} 0$ and $R(\delta B^{\textrm{crit}}) {\simeq} C^\infty$.%
Furthermore, we show in Sec.~\ref{sec:polynomial} that the {}{analysis of $C(B)$ with peaky signaling} in literature  \cite{Zheng2007noncoherent,Ray2007noncoherent}  {}{experiences exactly this} same gap to $C^\infty$, {}{although we obtained}~\eqref{eqn:low:peak} using non-peaky signals~\cite{journals/twc/LozanoP12}  and a power-boosting duty cycle $\delta{\in}(0,1]$. 
Fig.~\ref{fig:CBCBd} illustrates the relation between our bounds $C(B,\delta)$, capacity $C(B)$, and the coherent wideband channel limit $C^\infty$. 
\begin{figure}[t]
  \centering
  \includegraphics[width=.6\columnwidth]{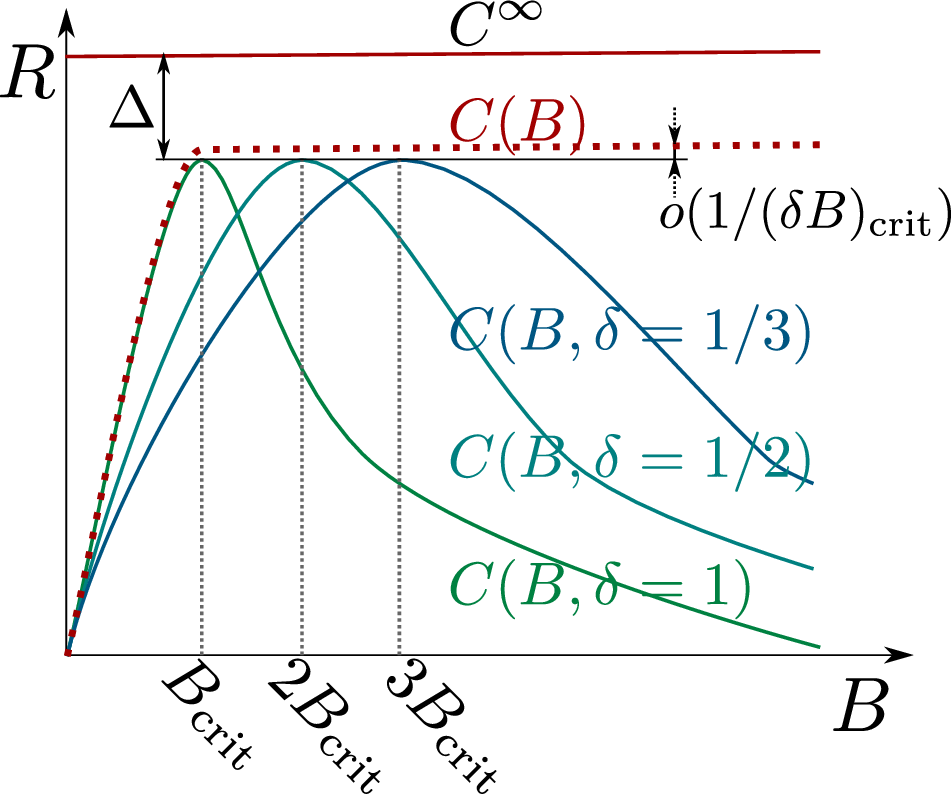}
  \caption{All transmission strategies with the same bandwidth occupancy $\delta B {=}(\delta B)_{\mathrm{crit}}$ achieve the same polynomial approximation of $C^\infty$ at different bandwidths. $C(B)$ is separated from {}{the maximum $C(B,\delta)$} by a difference of $o(1/(\delta B)_\mathrm{crit})$.
  }
  \label{fig:CBCBd}
  \end{figure}
 
The main contribution of this paper is the unified approximation of $C^\infty$ with peaky and non-peaky signaling, showing that these two extremes can be connected by all level of peakiness parametrized by the duty cycle $\delta{\in}(0,1]$. All signaling schemes $(B, \delta)$ with the same bandwidth occupancy $\delta B {=} (\delta B)_{\mathrm{crit}}$ approach  $C^\infty$ within the same capacity gap up to a small approximation error of $o(1/(\delta B)_{\mathrm{crit}})$.  We have also derived closed-from expressions for capacity bounds and critical bandwidth occupancy for all values $\delta{\in}(0,1]$, which provide valuable engineering insights and tools to quantify the resources needed to approach {}{$C^\infty$}. As a byproduct, we obtained a group of closed-form bounds {}{to} the range of $(\delta B)_{\mathrm{crit}}$ {}{that are implicit in the existing literature~\cite{Zheng2007noncoherent,Ray2007noncoherent}. These} parametric bounds can be {}{tuned} based on an \textit{accuracy-resolution tradeoff} to complement the range {}{identified} in our non-peaky signaling analysis.  

\subsection{Related Work}

The results in \cite{journals/tit/MedardG02} have been extended to signals with output fourth-order constraint \cite{Hajek2002a} or small input peakiness constraint \cite{Hajek2002b}. 
Telatar and Tse \cite{journals/tit/TelatarT00} related channel uncertainty to the number of resolvable independent paths, and showed that in a rich scattering environment where {}{this number} increases with $B$, the rate can grow as long as the signal power in each path is not too low, but it starts decreasing when the number of paths is above a critical value.

The capacity of a wideband fading channel {}{achieves} \textit{first order optimality} if, as $B$ goes to infinity, it has the same limit as a wideband additive white Gaussian noise (AWGN) channel. {}{This} has been studied in \cite{journals/tit/Verdu02,Medard2005,journals/tit/TelatarT00,Zheng2007noncoherent,Ray2007noncoherent} and {}{the linear in power capacity limit for MIMO  is}
\[\lim_{B\to\infty} C(B)^{\mathrm{noncoherent}}=  \lim_{B\to\infty} C(B)^{\mathrm{AWGN}} =\frac{N_{\mathrm{r}}P}{N_0}.\]

To quantify the ``exchange rate'' of bandwidth to capacity in the asymptotic regime where $B\to\infty$, the concept of \textit{wideband slope} was introduced in~\cite{journals/tit/Verdu02}. A larger wideband slope means that higher rate gain is obtained given the same amount of extra bandwidth. The wideband slope is {}{studied} in~\cite[Theorem 9]{journals/tit/Verdu02} based on the second order term of a Taylor series expansion of the spectral efficiency ($C/B$, in nats/s/Hz) with respect to the signal-to-noise  ratio (SNR) at each receive antenna, $\textrm{SNR}{\triangleq}P/(BN_0)$.  
The wideband slope is inversely proportional to the second order derivative of the spectral efficiency at $\textrm{SNR}{=}0$, which is finite for AWGN and coherent fading channels (i.e., with perfect CSIR) but $-\infty$  for non-coherent scenarios (i.e., with no CSIR). {}{Thus the coherent fading channel is \textit{second order optimal} but the non-coherent channel is not.} This abrupt distinction contrasts with the intuition that, as the channel coherence time $T_c$ and/or frequency $B_c$ grow, channel estimation becomes increasingly rewarding and the capacity of the non-coherent channel converges to the capacity of the coherent channel.   This contradiction was resolved in \cite{Zheng2007noncoherent,Ray2007noncoherent} by showing that in non-coherent Rayleigh fading channels the spectral efficiency  $C(B)/B$ is better represented by a generalized Taylor polynomial of order $1{+}\alpha<2$,
\begin{align}
\label{eqn:Ray}
\frac{C(B)}{B} {=}  N_{\mathrm{r}} \left(\frac{P}{N_0 B}\right) {-} \frac{N_{\mathrm{r}}(N_{\mathrm{r}}{+}N_{\mathrm{t}})}{2N_{\mathrm{t}}}\left(\frac{P}{N_0 B}\right)^{1+\alpha}+o(\frac{1}{B^{1+\alpha}}), [\mbox{nats/s/Hz}],
%\nonumber\\
% \frac{C(B)}{B}& = N_{\mathrm{r}} \mathrm{SNR}\left(1- \frac{ N_{\mathrm{r}}{+}N_{\mathrm{t}} }{2N_{\mathrm{t}}}\textrm{SNR}^{\alpha}\right)+o(B^{-(1+\alpha)}),
\end{align}
where the exponent $\alpha{\in}(0,1)$ grows with increasing $B_cT_c$. The first term equals $C^{\infty}/B$, representing a {}{\textit{first order optimal}} upper bound of the spectral efficiency when rate is power-limited. 
The third term captures the approximation error, that vanishes faster than $B^{-(1{+}\alpha)}$ as $B{\to}\infty$. 
The second term represents the penalty from lack of channel knowledge. It contains $\textrm{SNR}^{1{+}\alpha}$, a sub-quadratic term  ($1{+}\alpha{<}2$) that characterizes the convergence speed of the spectral efficiency for non-coherent fading channels. Representing \eqref{eqn:Ray} by the second order Taylor polynomial  leads to an infinite coefficient to the term $\textrm{SNR}^{2}$ {}{(wideband slope) and lack of second order optimality} as in~\cite{journals/tit/Verdu02}.
In this paper the word ``polynomial''  refers to these generalized Taylor polynomials with real-valued exponents.
 
Although peaky signals are imperative to achieve first order optimality \cite[Th. 7]{journals/tit/Verdu02}, they are challenging to synthesize owing to hardware non-linearity and the infinite amount of bandwidth they require in non-coherent channels. If a small gap from $C^\infty$ at a large but finite bandwidth is admissible, which is the case in all practical applications of asymptotic results, recent works have shown that non-peaky signals may suffice.
For example, Zhang and Laneman~\cite{Zhang07} investigated the achievable rate of phase-shift keying (PSK) for frequency-flat time-varying non-coherent Rayleigh fading channels. Under average power constraints, this signaling scheme approaches the wideband capacity limit for low but not too low SNRs. For signals subject to both peak and average power constraints, it was observed in~\cite{Wang09} that the gap between capacity upper and lower bounds can be very small for discrete-time frequency-flat Rayleigh fading channels. The capacity of non-coherent time-frequency selective wide-sense stationary uncorrelated scattering (WSSUS) channels with both peak and average power constraints has been studied in~\cite{Durisi}, where bell-shaped capacity upper and lower bounds were established and the capacity optimal bandwidth, the \textit{critical bandwidth},  was coarsely identified as a function of the peak power and the scattering function. For flat scattering functions, the capacity bounds depend on the system bandwidth and the input-signal peak constraint only through their ratio. The results in~\cite{Durisi} have been extended to MIMO in \cite{SDBP09}, where the impact of transmitter/receiver antenna correlation on capacity was also investigated. 
Lozano and Porrat \cite{journals/twc/LozanoP12} considered non-peaky signaling in SISO systems under a general fading distribution. Their results show that, when bandwidth is not too large, there is a transitory first stage where rate $R(B)$ grows with $B$ before approaching a maximum $R(B_{\textrm{crit}})$ at the critical bandwidth $B{=}B_{\textrm{crit}}$, beyond which the rate $R(B)$ decreases to zero as $B$ grows unbounded. By resorting to computation of mutual information rather than the capacity analysis as in~\cite{Durisi, SDBP09}, they provided closed-form expressions to the maximum rate and the corresponding capacity gap,  
 \begin{align}\label{eqn:Lozano}
   R(B_{\textrm{crit}}) = \frac{P}{N_0} -{\Delta},  \quad {\Delta}=\frac{P}{N_0}\sqrt{\frac{1+\log B_cT_c}{B_cT_c}\kappa\log\pi},
 \end{align}
 where  ${\Delta}$  vanishes with increasing coherence length $B_cT_c$. For Rayleigh fading, closed-form expressions for the range of $B_{\textrm{crit}}$ were also derived.

{}{Even though} \cite[Th. 7]{journals/tit/Verdu02} found that peaky signaling is {}{imperative} to achieve first order optimality, the definition of first order optimality {}{enforces an implicit} requirement to make bandwidth grow as high as possible ($B\to\infty$). {}{Thus, only those inputs that approach $C^\infty$ when $B$ is infinite are covered by \cite[Th. 7]{journals/tit/Verdu02}. What our results show is that $C^\infty$ can be approached as well} using a finite bandwidth $B$ and non-peaky signaling.

Unlike in~\cite{Zheng2007noncoherent,Ray2007noncoherent} where the non-coherent wideband fading channel capacity $C(B)$ is obtained by using the position of signal pulses in the frequency domain  (i.e., FSK) to convey information, in our analysis the position of actively transmitted symbols in the time domain, which collectively defines the active-cycle, is revealed  in advance to the receiver and therefore bears no information. 

Our  capacity bounds are based on computation of mutual information  with  constrained input signal peakiness -- in the sense of kurtosis -- that is controlled by enforcing a duty cycle $\delta\in(0, 1]$. This is in contrast to~\cite{Durisi, SDBP09} where capacity analysis is used with peak constraint on the amplitude of transmitting signals. Our choice of mutual information analysis can be justified from two aspects: even thought we do not design inputs to achieve the capacity bounds we can guarantee such inputs exist as long as the channel and noise are stationary weakly mixing processes, see \cite[Prop.~2.1]{Hajek2005}; the rate upper and lower bounds and the range of the critical bandwidth occupancy can be described in closed-form expressions, which are otherwise difficult to obtain using capacity analysis, see~\cite{Durisi, SDBP09}.

Our choice of using  duty cycle rather than peak constraint on signal amplitude~\cite{Wang09, Durisi} to control the signal peakiness can be justified as follows: given the same average power constraint, a peak constraint on signal amplitude will limit the peak-to-average power ratio (PAPR), which is sufficient but not necessary to generate a constraint on signal peakiness. Signals with finite peakiness may have infinite PAPR (e.g., Gaussian signal has infinite PAPR but only a small kurtosis $\kappa{=}2$). It must be noted that in non-coherent wideband fading channels, capacity is related to the peakiness in the kurtosis sense~\cite{journals/tit/MedardG02, journals/tit/Verdu02}.

The rest of this paper is organized as follows. We introduce the system model in Sec.~\ref{sec:model} and present our unified analysis of wideband non-coherent channel in Sec.~\ref{sec:Bdlimit}. We describe our non-coherent polynomial approximation to coherent capacity, and discuss its relation with literature in Sec.~\ref{sec:polynomial}. Finally our conclusions are in Sec.~\ref{sec:conclusion}.

\section{System Model}\label{sec:model}

We consider a rich scattering, frequency selective, block fading, $N_\mathrm{t}\times N_\mathrm{r}$ MIMO wideband channel with an impulse response $h(t)^{(u,v)}$ between antennas $(u,v)$. For compactness we assume that all channels experience a coherence time $T_c$ and a delay spread $D$ and the channel frequency response becomes uncorrelated for frequencies apart more than one coherence bandwidth $B_c{\triangleq}1/D$. We focus only on the frequency signaling scheme since it is known \cite{journals/twc/LozanoP12} that differences between frequency and time signaling only affect the scaling with bandwidth in its vanishing higher order terms. In the following we present the characteristics of the discrete-time system model\footnote{The equivalence between the discrete-time and continuous-time channel models for SISO is established in~\cite{TSP04} using sampling and DFT, and in~\cite{Durisi12} using pulse shaping filter banks with Weyl-Heisenberg projection. Our result uses MIMO in a rich scattering environment and we provide explicit mapping of the channel coefficients between two different discrete-time models.}. Justification of our choice of the wideband fading model is presented in Appendix~\ref{app:dvSysMod}.

Our model starts from a continuous-time wideband fading channel, followed by the discretization/sampling process on the input-output signals. This provides a signaling scheme where every $T_c$ seconds, the transmitted signal $x^{(u)}[n]$ with bandwidth $B$ carries $K{=}BT_c$ complex samples on antenna $u{\in}[0,N_\mathrm{t}{-}1]$. Taking a $K$-point DFT of the complex samples for each antenna and then stacking all the $N_\mathrm{t}$ vectors up, the transmitted codeword is uniquely defined by the $N_\mathrm{t}K\times 1$ vector $\x$ that satisfies the average power constraint
\[\frac{1}{K}\Ex{}{|\x|^2}\leq PT_c.\]
 For $i{=}k N_\mathrm{t}{+}u$, the $i$-th coefficient of $\x$, denoted as $x^{(i)}$, corresponds to the transmitted signal on antenna $u$ with DFT index $k{\in}\{0,1,\ldots, K{-}1\}$.
For each pair of antennas $(u,v)$, the discrete samples of the channel have $M{=}BD$ i.i.d. coefficients $h^{(u,v)}[n]$, $n{=}0,1,\ldots,M{-}1$, with $M/K{=}D/T_c{=}\frac{1}{B_cT_c}$.
After applying $K$-point DFT to each discrete channel sequence $h^{(u,v)}[n]$, we define a block-diagonal matrix  
\begin{equation}
 \Hb=\left(\begin{array}{c|c|c|c}
           \Hb[0]&\zero &\cdots&\zero\\\hline
            \zero&\Hb[1]&\ddots&\vdots\\\hline
           \vdots&\ddots&\ddots&\zero\\\hline
           \zero&\dots&\zero&\Hb[K-1]\\
          \end{array}\right),
 \label{eq:MatrixChannel}
\end{equation}
where $\Hb[k]$ contains in its $(v,u)$-th element the $k$-th DFT coefficient of $h^{(u,v)}[n]$, whose distribution is determined by the impulse response $h(t)^{(u,v)}$. Each channel only has $M$ i.i.d. coefficients and any two blocks $\Hb[k]$ and $\Hb[k']$ are correlated only if $|k{-}k'|{<}B_cT_c$. We also define the average gain of the $n$-th channel coefficient $g_{n}^{(u,v)}{=}\Ex{}{|h^{(u,v)}[n]|^2}$ satisfying $\sum_{n=0}^{M-1}g_n^{(u,v)}=1$.

When $D{\ll} T_c$, a cyclic prefix with negligible influence in rate can be inserted to remove the inter-symbol interference and the signal received on each fading realization, $T_c$, depends only on the state of the channel and signal transmitted during the same realization. After applying  $K$-point DFT to the received signal, we can represent the system as
\begin{equation}\label{eq:sysModel}
 \y=\Hb\x+\z,
\end{equation}
where $\y$ is a $N_\mathrm{r}K\times 1$ vector whose $i$-th element $y^{(i)}$, with $i=kN_\mathrm{r}{+}v$, corresponds to the signal received on antenna $v$ with  DFT coefficient index $k$. The noise vector $\z$ follows a Gaussian distribution with PSD $N_0$ ($\mathcal{CN}(0,\I_{N_\mathrm{r}K}N_0T_c)$).

Some references, such as \cite{Zheng2007noncoherent,Ray2007noncoherent}, use a different  discrete-time model with fewer frequency bins, each experiencing an independent fading coefficient that repeats itself for many consecutive symbols. We prove in Appendix~\ref{app:bkSysMod} that the two discrete-time models are compatible. In Appendix~\ref{app:bkSigRepr} we show that the two models are equivalent at the continuous-time level using concepts of multi-carrier modulations and we provide explicit mapping of the channel coefficients between the two models. Therefore our results are independent of the model chosen.

Wideband capacity is related to peakiness in the sense of the normalized fourth moment of the inputs, or \textit{kurtosis} \cite{journals/tit/MedardG02, journals/tit/Verdu02}.  
Given a stochastic sequence $A(t)$, its kurtosis is defined as
\begin{equation}
 \kappa(A(t))\triangleq\frac{\Ex{A(t)}{|a(t)|^4}}{\Ex{A(t)}{|a(t)|^2}^2},
\end{equation}
where the time index $(t)$ may be dropped if the process is stationary.   
By enforcing  a duty cycle $\delta{\in}(0,1]$ on the input signal $\x$, the system 
is converted into the time-alternation of an active stage for a fraction $\delta$ of the time with boosted power $P'{=}\frac{P}{\delta}$, and an idle stage for a fraction $(1{-}\delta)$ of the time. 
Let $\tilde{\x}$ be a non-peaky signal with power $P$ and finite kurtosis $\kappa(\tilde{\x})$. {}{We introduce a binary random variable $c{\in}\{0,1\}$ to represent the use of each fading block of size $ T_c{\times} B_c$,  where $c{=}1$ means the channel block is active for signal transmission and $c{=}0$  means idle, with probability $P_r(c{=}1)=\delta$. We reveal $c$ to the receiver in advance, which will reduce the rate as 
\[C(B,\delta)=\CInf{\x}{\y}{c}=\Inf{\x,c}{\y}-\Inf{c}{\y}= \Inf{\x}{\y} -\Inf{c}{\y} \leq \Inf{\x}{\y},\] 
where $0{\leq} \Inf{c}{\y} {\leq} \Ent{c}$ with all equalities hold for $\delta{=}1$.   The duty cycle} induces a new signal
\begin{equation}
\x {}{= \tilde{\x} \sqrt{\frac{c}{\Ex{}{c}}}}=\begin{cases}
              \tilde{\x}/\sqrt{\delta}, & \mbox{w.p. } \delta,\\
							 0, & \mbox{w.p. } 1{-}\delta,
							\end{cases}
							\mbox{ with }
	\kappa(\x){=}\frac{\Ex{}{|\x|^4}}{\Ex{}{|\x|^2}^2}= \frac{\Ex{}{|\tilde{\x}|^4}}{\delta \Ex{}{|\tilde{\x}|^2}^2} = \frac{\kappa(\tilde{\x})}{\delta}.						
\end{equation} 
Therefore we can effectively adjust the peakiness (in the sense of kurtosis) of signaling without imposing any extra constraint on the distribution of the active signal  $\tilde{x}$.  

\section{Bandwidth Occupancy Limit}\label{sec:Bdlimit}

Our analysis is a generalization of the the SISO analysis with non-peaky signaling in~\cite{journals/twc/LozanoP12}. We extend the process to MIMO systems and to an arbitrary level of signaling peakiness through the tunable duty cycle parameter $\delta\in(0,1]$. Both analyses follow four steps, represented in Fig.~\ref{fig:analisisref}.
\begin{enumerate}
 \item Find a bell-shaped lower bound $R^{LB}(\delta B)\leq C(B,\delta)$;
 \item Determine the unique maximum of $R^{LB}(\delta B)$, $R^{LB}((\delta B)^*)$;
 \item Find a bell-shaped upper bound $R^{UB}(\delta B)\geq C(B,\delta)$;
 \item Determine  $(\delta B)^+$ and $(\delta B)^-$ such that  
$R^{UB}((\delta B)^+){=}R^{UB}((\delta B)^-){=}R^{LB}((\delta B)^*)$.
\end{enumerate}
The result of \cite{journals/twc/LozanoP12} shows that the capacity of a non-coherent fading channel with non-peaky signaling ($\delta{=}1$, finite $\kappa$) grows with bandwidth $B$ only when it is below a \textit{critical bandwidth} $B_{\mathrm{crit}}$, which falls into the range $[B^-,B^+]$. A system operating with insufficient bandwidth $B{<} B_{\mathrm{crit}}$ is less efficient in converting available signal energy into rate due to the sub-linear law between rate and SNR, and the corresponding achievable rate grows with increasing bandwidth.
When signal power spreads over too much bandwidth $B{>}B_{\mathrm{crit}}$, the channel-uncertainty induced penalty grows with increasing bandwidth and the achievable rate decreases to zero as $B{\to}\infty$. Therefore, contrary to the wideband AWGN channel where ``the deeper into the low-SNR
regime, the better'', in the non-coherent fading channel the  guideline is ``enter, but not in excess, in the low-SNR regime'', with the optimal operation point at $B_{\mathrm{crit}}$. Our result shows that for any $B{>}B_{\mathrm{crit}}$ it is possible to bring the capacity back to the {same} optimal value, up to a small approximation error of order $o(1/B_{\mathrm{crit}})$, by imposing a duty-cycle parameter $\delta{=} (\delta B)_{\mathrm{crit}}/B$ and a power-boost $P'=P/\delta$ on the original non-peaky signaling. Moreover, in Sec.~\ref{sec:polynomial} we show that this strategy achieves the same gap from $C^\infty$ as in the peaky-signaling analysis \cite{Zheng2007noncoherent,Ray2007noncoherent}.

\begin{figure}[t]
\centering
\includegraphics[width=.7\columnwidth]{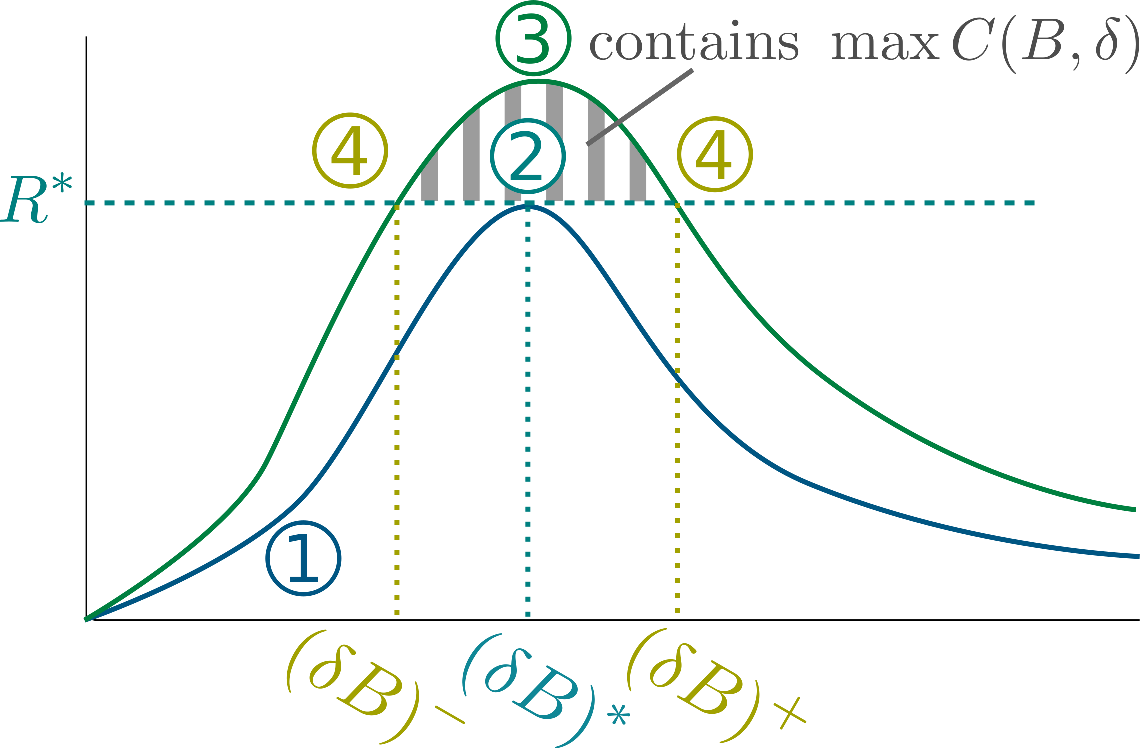}
 \caption{Our four-step analysis of critical bandwidth occupancy. Substituting $(\delta B)$ by $B$ gives the original analysis by \cite{journals/twc/LozanoP12}.}
 \label{fig:analisisref}
\end{figure}

% Already mentioned in Sec.~\ref{sec:introduction}.
%This method uses mutual information bounds, rather than capacity bounds as in for instance [16]. The connection between mutual information and capacity bounds is well established in [18,Prop.~2.1]. The advantege of this analysis is that characterizes rate for any proper complex input constellation, and closed form expressions for the critical bandwidth are obtained. This allows to give quantitative values to the trends earlier identified in [16].

\subsection{Capacity Lower Bound for $C(B,\delta)$}

As in \cite{journals/twc/LozanoP12}, our lower bound is obtained by first calculating the maximum achievable rate of a coherent non-peaky signaling under the average power constraint and then deducting the maximum rate  penalty from lack of CSI. Therefore it is valid for general channel fading distributions and for any value of $B$ and $B_cT_c{>}N_t$. 
The potential spatial correlation among different antennas is not considered here.

\begin{lemma}
\label{lem:RLB}
The achievable rate in a wideband  non-coherent channel with i.i.d. fading and a duty cycle $\delta\in(0,1]$ is lower bounded by
\begin{equation}
\label{eq:RLB}
R^{LB}(\delta B) = \frac{PN_{\mathrm{r}}}{N_0}\left[1-\frac{P(\kappa-2+N_{\mathrm{t}}+N_{\mathrm{r}})}{2\delta BN_{\mathrm{t}}N_0}\right]
  -\frac{\delta BN_{\mathrm{t}}N_{\mathrm{r}}}{B_cT_c}\log\left(1+\frac{P}{\delta BN_{\mathrm{t}}N_0}B_cT_c\right),
\end{equation}
where $\kappa = \kappa(h)$ is the \emph{kurtosis} of the channel fading coefficients.
\end{lemma}
\begin{proof}
See Appendix \ref{app:Prooflem1}.
\end{proof}

The \emph{kurtosis} $\kappa$ for many fading distributions are in the range of $[1,2]$. For example, as given in~\cite{journals/twc/LozanoP12}, $\kappa=2$ for Rayleigh fading, $\kappa=2{-}4k^2/(1{+}2k)^2$ for Rice fading with factor $k>0$, and $\kappa=1{+}1/m$ for Nakagami-$m$ fading channels.

\begin{remark}
Even though the duty-cycle constrained capacity $C(B,\delta)$ might be two-dimensional function of $\delta$ and $B$ , the lower bound is only a function of the product $\delta B$. 
\end{remark}

\subsection{Maximum of $R^{LB}$}

We use the assumption that $B_cT_c{\gg}1$ to determine the maximum of the capacity lower bound $R^{LB}$. For any finite $B_cT_c$ we can  approximate the optimal $(\delta B)^*$ and the associated maximum rate up to an error term $o\left(\sqrt{\frac{\log (B_cT_c)}{B_cT_c}}\right)$ that decreases to zero as $B_cT_c\to\infty$.

\begin{lemma}
\label{lem:optimB}
$R^{LB}(\delta B)$ is maximized at $\delta B{=}(\delta B)^*$ with
\begin{align}
 (\delta B)^* & = \frac{P}{N_0N_{\mathrm{t}}}\sqrt{\frac{B_cT_c}{\log (B_cT_c)}(\kappa-2+N_{\mathrm{t}}+N_{\mathrm{r}})}+o\left(\sqrt{\frac{B_cT_c}{\log (B_cT_c)}}\right), \label{eq:optimB}\\
 R^{LB}((\delta B)^*) & \geq \frac{PN_{\mathrm{r}}}{N_0}\left[1-\sqrt{\frac{1+\log (B_cT_c)}{B_cT_c}(\kappa{-}2{+}N_{\mathrm{t}}{+}N_{\mathrm{r}})\log\pi}\right]-o\left(\sqrt{\frac{\log (B_cT_c)}{B_cT_c}}\right).\label{eq:optimR}
\end{align}
\end{lemma}
\begin{proof}
See Appendix \ref{app:Prooflem2}.
\end{proof}

\begin{remark}
We would like to emphasize that the rate-maximizing bandwidth occupancy $(\delta B)^*$ is very large given the fact that the channel coherence $B_cT_c$ usually ranges from a few hundreds to hundreds of thousands. For example, assuming $2\times 2$ MIMO over Rayleigh fading ($\kappa{=}2$) with $P/N_0{=}70$ dB, we have $(\delta B)^*{\simeq} 120$ MHz with capacity gap $\Delta/C^\infty{<}0.18$ for $B_cT_c{=}10^3$, and $(\delta B)^*{\simeq} 930$ MHz with $\Delta/C^\infty{<}0.03$ for $B_cT_c{=}10^5$.
\end{remark}

\subsection{Capacity Upper Bound for $C(B,\delta)$}

We obtain a capacity upper bound for the case when channel is Rayleigh distributed. The bound, up to an  error term of $o(1/\delta B)$ that vanishes as $\delta B\to\infty$, applies to any value of $B$ and $B_cT_c{>}N_t$ and  all inputs subject to constraints of  average power $P$ and signaling duty cycle $\delta$.

\begin{lemma}
\label{lem:RUB}
The achievable rate of  signaling schemes with duty cycle $\delta{\in}(0,1]$ in a wideband non-coherent Rayleigh fading channel is upper bounded by
\begin{align}\label{eq:RUB}
  R^{UB}(\delta B)=\frac{PN_{\mathrm{r}}}{N_0}\Bigg[1-\frac{P}{2\delta BN_0}
  -\frac{\delta BN_{\mathrm{t}}N_0}{PB_cT_c}\Ex{\psi}{\log(1{+}\frac{P}{\delta N_{\mathrm{t}} BN_0}B_cT_c g_{\min} \psi)}\Bigg]+o(\frac{1}{\delta B}),
\end{align}
where $g_{\min}{=}\min_{m,u,v} \Ex{}{|h^{(u,v)}[m]|^2}$ is the minimum non-zero square channel gain among all delays and antenna pairs, 
and the random variable $\psi$   is defined as 
\begin{equation}
\label{eq:psidef}
\psi{=}\arg\min_{\psi_{K,n}} \Ex{}{\log\left(1{+}\frac{P g_{\min}\psi_{K,n}}{\delta N_{\mathrm{t}} N_0B}B_cT_c\right)},\
\mbox{ where }
\textstyle{}\psi_{K,n}{\triangleq}\frac{1}{K}|\sum_{k=0}^{K-1} \frac{\x_k}{\sqrt{P}}e^{-j2\pi \frac{kn}{MN_{\mathrm{t}}}}|^2.
\end{equation}
\end{lemma}
\begin{proof}
See Appendix \ref{app:Prooflem3}.
\end{proof}

\begin{remark}
 The auxiliary variable $\psi$ is bounded ($\psi{>}0$, $\Ex{}{\psi}{\leq}1$) and serves here as a placeholder for the minimization of the last term of the bound, which is implicitly determined by \eqref{eq:psidef}.
\end{remark}

\subsection{Critical Bandwidth Occupancy}
We obtain the range of values of $\delta B$ where the upper bound is larger than $R^{LB}((\delta B)^*)$. $C(B,\delta)$ can approach $C^\infty$ within the small gap in \eqref{eq:optimR}  only if the bandwidth occupancy is contained in an interval that grows linearly with $\sqrt{\frac{B_cT_c}{\log (B_cT_c)}}$, as suggested by \eqref{eq:optimB}, and the error term $o(\frac{1}{\delta B})$ in Lemma \ref{lem:RUB} can be substituted with an equivalent term $o\big(\sqrt{\frac{\log (B_cT_c)}{B_cT_c}}\big)$.
\begin{lemma}
\label{lem:critB}
In a wideband  non-coherent Rayleigh fading channel, the maximum rate in \eqref{eq:optimR} is achievable at a critical bandwidth occupancy $(\delta B)_\mathrm{crit}$ that resides in the range
\begin{equation}
 (\delta B)^- \leq (\delta B)_\mathrm{crit}\leq (\delta B)^+,
\end{equation}
where
\begin{equation}\label{eq:dBcritical}
\begin{split}
(\delta B)^-  &{=}\frac{P}{N_0}\frac{1}{ 2\sqrt{(N_{\mathrm{t}}+N_{\mathrm{r}})\log\pi}}\sqrt{\frac{B_cT_c}{\log (B_cT_c)}}+o\left(\sqrt{\frac{B_cT_c}{\log (B_cT_c)}}\right),\\
(\delta B)^+ &{=}\frac{P}{N_0}2\sqrt{\frac{(N_{\mathrm{t}}+N_{\mathrm{r}})}{N_{\mathrm{t}}^2}\log\pi}\sqrt{\frac{B_cT_c}{\log (B_cT_c)}}+o\left(\sqrt{\frac{B_cT_c}{\log (B_cT_c)}}\right).
\end{split}
 \end{equation}
\end{lemma}
\begin{proof}
See Appendix \ref{app:Prooflem4}.
\end{proof}

\subsection{Interpretation of the Result}

Our upper and lower bounds on $C(B,\delta)$ are all derived from the chain rule  
$${\CInf{\x}{\y}{c}} =  {}{\delta}\CInf{\x,\Hb}{\y}{{}{c{=}1}} - {}{\delta}\CInf{\Hb}{\y}{\x{}{,c{=}1}},$$ 
where the first term  corresponds to the data transmission setup
that quantifies the information about $\Hb\x$ contained in $\y$ and the second term can be interpreted as a ``channel estimation'' setup that quantifies the rate penalty for not knowing $\Hb$. Both terms grow as $\delta B$ increases but the first term grows faster when $\delta B$ is small, thus increases ${}{\CInf{\x}{\y}{c}}$, until the second term ``accelerates''. Beyond the critical point $(\delta B)_{\mathrm{crit}}$, the second term grows faster than the first, thus erodes ${}{\CInf{\x}{\y}{c}}$, until the capacity drops to zero when $\delta B\to\infty$. 
This behavior is illustrated in Fig. \ref{fig:analisisref} for both the capacity upper and lower bounds.
The coherence time $T_c$ and coherence bandwidth $B_c$ of the channel jointly determine the relative speeds of this ``race'' through their product. 
% $T_cB_c$: the average number of independent signal realizations per each independent channel realization. 
The factor $T_cB_c$ appears in rate penalty both as the denominator outside the logarithm (there are $T_cB_c$ times	fewer i.i.d. channel realizations than signal realizations) and as a multiplier of the SNR inside the logarithm (the power of  $T_cB_c$ signal realizations can be combined to estimate each channel realization), {}{leading to a capacity gap} depending on $\frac{\log B_cT_c}{B_cT_c}$.

\begin{figure}[!t]
 \centering\subfigure[Upper bound for $(\delta,B)$ and the lower bound for $\delta{=}1$.]{
 \includegraphics[width=.65\columnwidth]{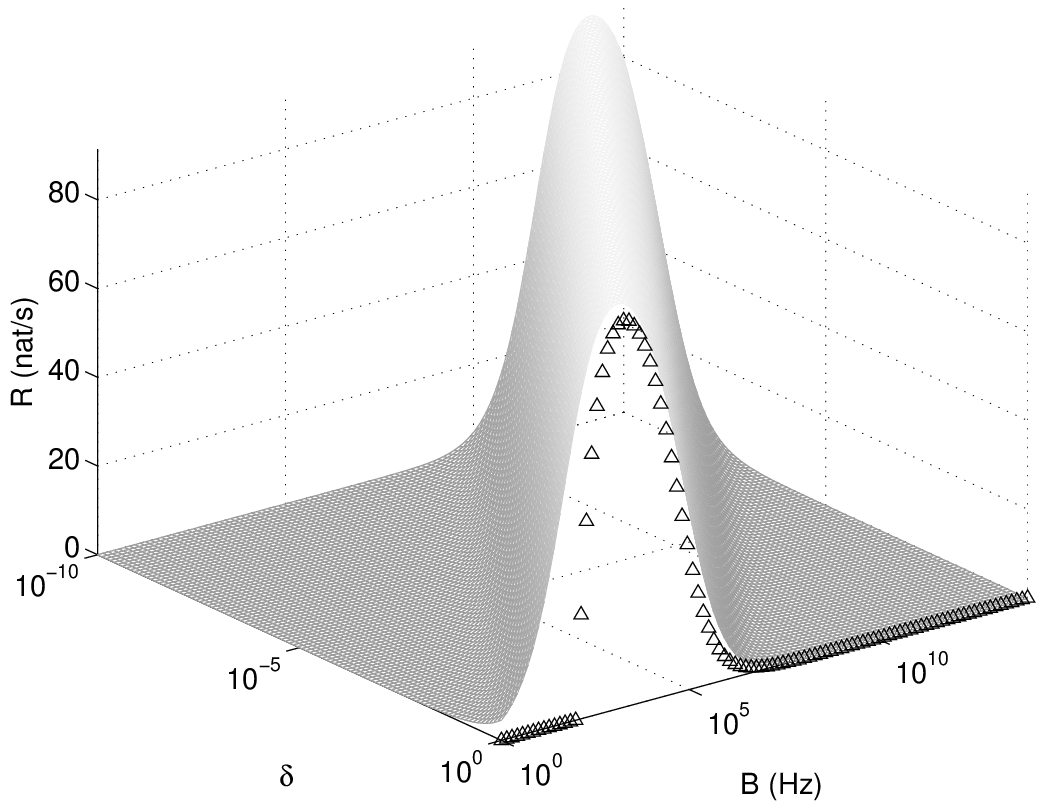}
 \label{fig:3DCapInterpretation}
 }
 \centering\subfigure[Contour plot  and levels of bandwidth occupancy.]{
 \includegraphics[width=.65\columnwidth]{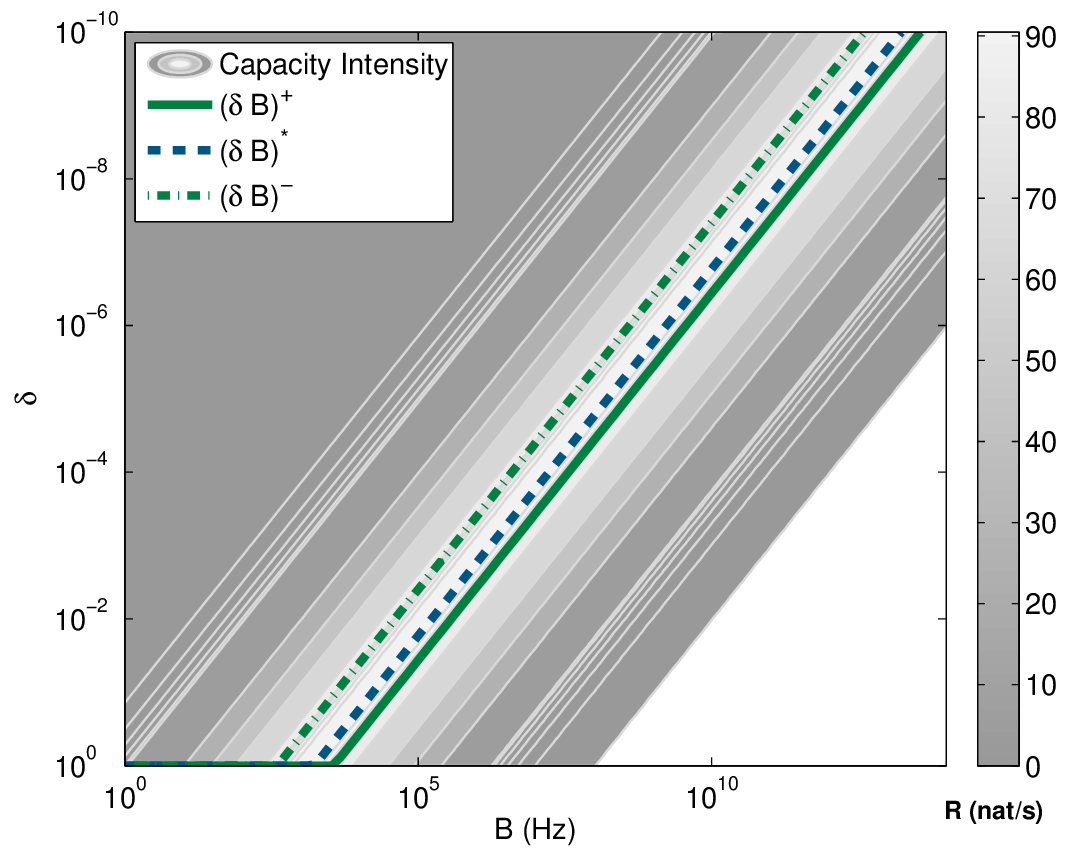}
 \label{fig:2DCapInterpretation}
 }
 \caption{Capacity upper bound over the plane $(\delta, B)$ and the low bound for $\delta{=}1$, with $B_cT_c=10^3$, $N_{\mathrm{t}}{=}N_{\mathrm{r}}{=}1$, and an intentionally chosen small value $P/N_0{=}20$dB.
 Range of critical bandwidth occupancy is also shown.}
  \label{fig:CapInterpretation}
\end{figure}

In Fig.~\ref{fig:3DCapInterpretation} we represent the upper bound to capacity as a field over the 2D plane $(\delta, B)$ with $B_cT_c{=}10^3$ and $P/N_0{=}20$ dB. In the vertical cut for $\delta{=}1$ we  also plot the lower bound using triangular bullets. Note that we intentionally choose a smaller value of $P/N_0$ to illustrate the details of the transition phase in capacity, which would otherwise be difficult to observe with typical values of $P/N_0{\sim}70$ dB~\cite{journals/twc/LozanoP12}. On the $B$-axis, we can see that for fixed values of $\delta$ the capacity as a function of bandwidth is bell-shaped, grows at small bandwidth, reaches a maximum and then decreases to zero.
Fig.~\ref{fig:2DCapInterpretation} provides a better perspective on the value of capacity upper bounds as a function of the bandwidth occupancy $\delta B$, where the optimal $(\delta B)^*$ that maximizes the capacity lower bound $R^{LB}$ and the range $[(\delta B)^-, (\delta B)^+]$ for the critical bandwidth occupancy $(\delta B)_{\mathrm{crit}}$ are also plotted. 
For different level of peakiness $\delta$, the peak values of capacity are the same but appear at different bandwidth $B$, and in fact all points with identical value $\delta B$ have the same lower/upper bounds.
Our analysis recovers previous results for non-peaky signals by setting $\delta{=}1$, producing a finite critical bandwidth. It also recovers the capacity with infinite-fourth-moment signals by taking $\delta{\rightarrow}0$, which drives the critical bandwidth occupancy point further into higher bandwidths satisfying $\displaystyle \lim_{\delta\to0}\frac{(\delta B)_{\mathrm{crit}}}{\delta}=\infty$.

Our analysis also unveils the impact of the dimensions of the MIMO array on the non-coherent wideband fading channels. The maximum rate  in Lemma~\ref{lem:optimB}, derived under the condition that $N_t{<}B_cT_c$, depends critically on channel coherence $B_cT_c$.. For example, for $B_cT_c=10^3$ and $B_cT_c=10^6$, the maximum rate \eqref{eq:optimR} can be approximated  as, respectively,
 \[R(B_cT_c{=}10^3)\simeq \frac{PN_r}{N_0}(1-0.1\sqrt{N_t{+}N_r}), \ R(B_cT_c{=}10^6)\simeq \frac{PN_r}{N_0}(1-0.005\sqrt{N_t{+}N_r}).\]
When $N_r{>}1$ is fixed, increasing the number of transmit antennas will degrade the rate, with the gap growing linearly with $\sqrt{N_t{+}N_r}$. When $N_t{>}1$ is fixed and channel is relatively flat (hence $B_cT_c$ is large), the rate gap is negligible for typical MIMO setups and therefore the rate grows almost linearly with $N_r$. When the channel is rather dispersive (hence  $B_cT_c$ is small), however, increasing $N_r$ will produce a power gain that increases the rate at speed $N_r$ but at the same time will bring in more channel uncertainty that increase the penalty at rate proportional to $\sqrt{N_t{+}N_r}$. Therefore using too many receive antennas will hurt the achievable rate. For example, the maximum rate peaks around $N_r{=}40$ for $B_cT_c=10^3$. 
It must be noted that our analysis is accurate for conventional MIMO systems with $N_t<B_cT_c$, and extension to high-dimensional MIMO is out of the scope of the current paper.

\section{Unified Capacity for Peaky and non-Peaky Signals}\label{sec:polynomial} 

In this section we will show that the peak rate $R^{LB}((\delta B)^*)$ in~\eqref{eq:optimR}, which is derived  by combining non-peaky signaling analysis~\cite{journals/twc/LozanoP12} and tunable peakiness through duty cycle $\delta{\in}(0,1]$, approaches  $C^\infty$ within the same gap as in the unconstrained capacity $C(B)$ analysis using a generalized polynomial rate approximation \eqref{eqn:Ray} obtained via peaky signaling analysis~\cite{Zheng2007noncoherent,Ray2007noncoherent}. 
We first replace $\delta$, a free parameter in our model representing the duty cycle, with an assigned value $\delta{=}\mathrm{SNR}^{1-\alpha}$ as in \cite{Ray2007noncoherent},  where $\mathrm{SNR}\triangleq P/(BN_0)$ and $\alpha$ is the exponent that determines the wideband slope. This substitution will show that, when the channel coherence length is large, i.e., $B_cT_c\gg1$, the gaps to $C^\infty$ in \eqref{eq:optimR}  
and in \eqref{eqn:Ray} have the same value at points $(B,\delta=\mathrm{SNR}^{1-\alpha})$. Furthermore, we show that the sufficient and necessary conditions on the coherence length $B_cT_c$ to approach $C^\infty$, proved in \cite[Th. 1-Th. 3]{Ray2007noncoherent}, can also be established using our results. Once we have established that the results are equivalent, the opposite path can be taken and use the values of $\alpha$ obtained in \cite{Ray2007noncoherent} to calculate a new range of the critical bandwidth occupancy in closed-form expressions. We discuss the relationship of the two expressions, which have minor differences in the error terms of the calculation of $\alpha$, reveal a trade-off between accuracy and resolution in \cite{Ray2007noncoherent}, and demonstrate that the two methods represent the same optimal rate.

\subsection{Different Analyses Show the Same Results}

The analysis in \cite{Ray2007noncoherent} obtains a necessary and sufficient condition on the coherence length of the channel, $B_cT_c$,
to guarantee that capacity is above a polynomial of  $\mathrm{SNR}{=}\frac{P}{N_0B}$ as $B{\to}\infty$ with specified peakyness $\delta=\mathrm{SNR}^{\alpha-1}$. This result is given in \cite[Th. 3]{Ray2007noncoherent}, which is rewritten in the next lemma for easy reference.  The result is valid for arbitrary $B_cT_c$, but the necessary condition to approximate $C^\infty$  is akin to requiring that $B_cT_c$ be large.
\begin{lemma}[Th. 3 \cite{Ray2007noncoherent}]
\label{lem:ray}
For any $\alpha{\in}(0,1]$ and $\epsilon{\in}(0,\alpha)$ the capacity of a Rayleigh block-fading MIMO channel with coherence time $T_c$, coherence bandwidth $B_c$, and average signal to noise ratio $\mathrm{SNR}{=}\frac{P}{BN_0}$ is
\begin{equation}\label{eqn:ray}
  \frac{C(B)}{B}\geq N_\mathrm{r}\mathrm{SNR}-\frac{N_\mathrm{r}(N_\mathrm{r}+N_\mathrm{t})}{2N_\mathrm{t}}\mathrm{SNR}^{1+\alpha}+\Theta(\mathrm{SNR}^{1+\alpha+\epsilon}),
\end{equation}
if and only if there exists a $\sigma{\in}(0,\epsilon)$ such that
\begin{equation}\label{eqn:llc}
 B_cT_c=\frac{N_{\mathrm{t}}^2}{(N_\mathrm{r}+N_\mathrm{t})^2}\mathrm{SNR}^{-2(\sigma+\alpha)}.
\end{equation}
\end{lemma}

Recall that in Sec.~\ref{sec:Bdlimit}, our rate lower bound in \eqref{eq:RLB} contains three terms, the wideband capacity $C^\infty$, a non-linear rate penalty due to $\log(1 + \mathrm{SNR})$, and a rate penalty due to lack of CSIR. Below the optimal bandwidth occupancy $(\delta B)^*$, the third term of \eqref{eq:RLB} is smaller in absolute value than the second. Replacing the third term by the second term and substituting $\delta {=}\mathrm{SNR}^{1-\alpha}$, $\alpha{\in}(0,1)$ into \eqref{eq:RLB} produces the following sufficient condition in terms of the bandwidth occupancy $\delta B$, as stated in Corollary~\ref{cor:altRayTh2}.

\begin{corollary}
\label{cor:altRayTh2}
If $\delta B \leq (\delta B)^*$, the achievable rate is lower bounded by
\begin{equation}
 \label{eq:altRayLB}
 C(B,\delta)\geq\frac{PN_{\mathrm{r}}}{N_0}\left[1-\left(\frac{P}{BN_0}\right)^\alpha\frac{(\kappa{-}2{+}N_{\mathrm{t}}{+}N_{\mathrm{r}})}{N_{\mathrm{t}}}\right].
\end{equation}
\end{corollary}

On the other hand, above $(\delta B)_\mathrm{crit}$, the third term of \eqref{eq:RLB} is greater than the second. This means that $C(B,\delta)$ is smaller than \eqref{eq:altRayLB}, which leads to the necessary condition in Corollary~\ref{cor:altRayTh1}.

\begin{corollary}
\label{cor:altRayTh1}
In Rayleigh fading ($\kappa{=}2$), if $\delta {=}\mathrm{SNR}^{1-\alpha}$ and if
\begin{equation}
C(B,\delta)\geq\frac{PN_{\mathrm{r}}}{N_0}\left[1-\left(\frac{P}{BN_0}\right)^\alpha\frac{(N_{\mathrm{t}}+N_{\mathrm{r}})}{N_{\mathrm{t}}}\right],
\end{equation}
 then the bandwidth occupancy satisfies $\delta B<(\delta B)^+$.
\end{corollary}

Now we can use the necessary condition in Corollary~\ref{cor:altRayTh1} and the sufficient condition in Corrollary~\ref{cor:altRayTh2} on bandwidth occupancy $\delta B$  to prove the sufficient and necessary condition \eqref{eqn:llc}.
\begin{proposition}
\label{pro:RayTh2}
Corollary~\ref{cor:altRayTh2} implies the sufficient condition \eqref{eqn:llc} for Lemma~\ref{lem:ray}.
\end{proposition}
\begin{proof}
Substituting $\delta{=}\mathrm{SNR}^{1-\alpha}$ and $\kappa{=}2$ into \eqref{eq:optimB}, we can rewrite $\delta B{<}(\delta B)^*$ in Corollary~\ref{cor:altRayTh2} as
\begin{align}
B_cT_c>\frac{N_{\mathrm{t}}^2}{(N_{\mathrm{r}}+N_{\mathrm{t}})^2}\mathrm{SNR}^{-2\alpha}(N_{\mathrm{r}}+N_{\mathrm{t}})\log(B_cT_c).\nonumber
\end{align}
iSnce $(N_{\mathrm{r}}+N_{\mathrm{t}})\log(B_cT_c)$ is a constant and $B\rightarrow\infty$, we have $(N_{\mathrm{r}}+N_{\mathrm{t}})\log(B_cT_c)\leq \mathrm{SNR}^{-2\epsilon}$ for any $\epsilon>0$. Therefore, it is also sufficient to have
\[B_cT_c \geq \frac{N_{\mathrm{t}}^2}{(N_{\mathrm{r}}+N_{\mathrm{t}})^2}\mathrm{SNR}^{-2(\alpha+\epsilon)},\]
which is a sufficient condition that Lemma~\ref{lem:ray} transforms in the upper limit of $\sigma\leq\epsilon$.
\end{proof}

\begin{proposition}
\label{pro:RayTh1}
Corollary~\ref{cor:altRayTh1} implies the necessary condition \eqref{eqn:llc} for Lemma~\ref{lem:ray}.
\end{proposition}
\begin{proof} The necessary condition $\delta B<(\delta B)^+$ in Corollary~\ref{cor:altRayTh1} can be rewritten as 
\[B_cT_c>\frac{N_{\mathrm{t}}^2}{(N_{\mathrm{r}}+N_{\mathrm{t}})^2}\mathrm{SNR}^{-2\alpha }\frac{(N_{\mathrm{r}}{+}N_{\mathrm{t}})}{4\log\pi}\log(B_cT_c).\]
Therefore we can express the necessary condition that Lemma~\ref{lem:ray} sets as lower limit of $\sigma\geq0$, 
\[B_cT_c >\frac{N_{\mathrm{t}}^2}{(N_{\mathrm{r}}+N_{\mathrm{t}})^2}\mathrm{SNR}^{-2\alpha },
\]
as long as $\frac{(N_{\mathrm{r}}{+}N_{\mathrm{t}})}{4\log\pi}\log(B_cT_c)\geq 1$, i.e., $B_cT_c{\geq} \pi^{4/(N_{\mathrm{r}}+N_{\mathrm{t}})}$, which is always satisfied in wideband fading channels where $B_cT_c$ is very large, and thus $B_cT_c> \pi^2$.
\end{proof}

\begin{remark}
From Proposition~\ref{pro:RayTh2} and Proposition~\ref{pro:RayTh1}, it is not surprising that the power gain term $\log(B_cT_c)$ was lost in \cite{Ray2007noncoherent}, because this sub-polynomial variation of the result has been ``buried'' in the range of valid exponents $\epsilon$ of the error term $O(\mathrm{SNR}^{1+\alpha+\epsilon})$.
\end{remark}

\subsection{New Bounds on $(\delta B)_\mathrm{crit}$ using the Subquadratic Polynomial Rate Approximation}

In our analysis, Lemma~\ref{lem:optimB} prescribes a near-linear-in-power capacity lower bound which can be achieved by all signaling schemes  with $(\delta, B)$  as long as the bandwidth occupancy $\delta B$ equals some constant  $(\delta B)_\mathrm{crit}$. Our analysis does not provide the exact value of $(\delta B)_\mathrm{crit}$, but rather bounds it within $[(\delta B)^-, (\delta B)^+]$ in Lemma~\ref{lem:critB}.
On the other hand, the result in \cite[Th.~3]{Ray2007noncoherent}, reproduced here as Lemma \ref{lem:ray}, prescribes an entire family of parametrized bounds where the parameter $\epsilon$ controls both the error term of the generalized Taylor expansion and the resolution of bounding brackets around $(\delta B)_\mathrm{crit}$. Corollary~\ref{cor:rayB} makes this explicit.

\begin{corollary}\label{cor:rayB}
The necessary and sufficient condition \eqref{eqn:llc} of Lemma~\ref{lem:ray} shows that coherent capacity $C^\infty$ is approached by transmitting signals with bandwidth occupancy $\delta B$ within the limits
\begin{align}
 \delta B & <\frac{P}{N_0}\frac{N_\mathrm{r}+N_\mathrm{t}}{N_\mathrm{t}}\sqrt{B_cT_c}\triangleq (\delta B)^{\max},\label{eqn:cor:bdupp} \\
\delta B & >\frac{P}{N_0}\left(\frac{N_\mathrm{r}+N_\mathrm{t}}{N_\mathrm{t}}\sqrt{B_cT_c}\right)^{\frac{\alpha}{\alpha+\epsilon}} \triangleq (\delta B)^{\min}_{\epsilon}.\label{eqn:cor:bdlow}
\end{align}
\end{corollary}
\begin{proof}
Substituting $B_cT_c$, $\delta{=}\mathrm{SNR}^{1-\alpha}$ and $\mathrm{SNR}{=}\frac{P}{BN_0}$ into \eqref{eqn:llc}, we can obtain \eqref{eqn:cor:bdupp} and \eqref{eqn:cor:bdlow} by the fact that $\sigma{>}0$ and $\sigma{<}\epsilon$, respectively.
\end{proof}

Therefore for a given $\alpha$, which controls the level of peakiness $\delta$ and determines the wideband slope, we can observe a clear tradeoff, parametrized by $\epsilon{\in}(0,\alpha)$, between the accuracy of the Taylor polynomial and the 
resolution of the bandwidth brackets:
\begin{enumerate}
\item The \textbf{accuracy} of the capacity lower bound calculated in Lemma~\ref{lem:ray} is determined by the ratio between  $\mathrm{SNR}^{1+\alpha}$ and the error term $O(\mathrm{SNR}^{1+\alpha+\epsilon})$. The larger $\epsilon$, the better the approximation, since the error term will vanish faster as $B\to\infty$.
\item The \textbf{resolution} of the interval where $(\delta B)_{\mathrm{crit}}$ is contained, $[(\delta B)^{\min}_{\epsilon},(\delta B)^{\max}]$, is determined by the width of the interval. The smaller $\epsilon$, the better the resolution, as the lower boundary $(\delta B)^{\min}_{\epsilon}$ will increase and become tighter when $\epsilon$ becomes smaller.
\end{enumerate}

\subsection{Comparison of Critical Bandwidth Occupancy Estimators}

So far we have characterized the critical bandwidth occupancy $(\delta B)_\mathrm{crit}$ in two different ranges: by a pair of bracket $[(\delta B)^-, (\delta B)^+]$, in our analysis in Lemma~\ref{lem:critB}; and by a parametric interval $[(\delta B)^{\min}_{\epsilon},(\delta B)^{\max}]$, derived from \cite[Th.~3]{Ray2007noncoherent} in
Corollary~\ref{cor:rayB}. 
To explore the relationship between the two estimators, we compare the difference in the estimated value of $\alpha$ that each analysis produces. We do this because the exponent $\alpha$ provides a unique relation between $B$ and $\delta{=}\mathrm{SNR}^{1-\alpha}{=}(\frac{P}{N_0B})^{1-\alpha}$, allowing for scalar comparison of the methods.

We begin by representing $\alpha$ according to \cite[Th.~3]{Ray2007noncoherent}. From \eqref{eqn:llc} in Lemma \ref{lem:ray}, for given values of the coherence block length $B_cT_c$ and bandwidth $B \in[(\delta B)^{\min}_{\epsilon},(\delta B)^{\max}]$ can be written as
\begin{equation}\label{eqn:AlphaSigma}
\sigma+\alpha = \frac{\log(\frac{(N_\mathrm{t}+N_\mathrm{r})^2}{N_\mathrm{t}^2}B_cT_c)}{2\log(\mathrm{SNR}^{-1})}.
\end{equation}
From the fact that $\sigma>0$ we get 
\begin{equation}
\label{eq:amaxRay}
 \alpha<\alpha_{\max}\triangleq\frac{\log(\frac{(N_\mathrm{t}+N_\mathrm{r})^2}{N_\mathrm{t}^2}B_cT_c)}{2\log(\mathrm{SNR}^{-1})},
\end{equation}
and from the fact that $\sigma<\epsilon<\alpha$ we get 
\begin{align}
\alpha> \max \left(\frac{\alpha_{\max}}{2}, \ \alpha_{\min}(\epsilon)\right), % \nonumber \\
\ \mbox{where} \
\alpha_{\min}(\epsilon)\triangleq\frac{\log(\frac{(N_\mathrm{t}+N_\mathrm{r})^2}{N_\mathrm{t}^2}B_cT_c)}{2\log(\mathrm{SNR}^{-1})}-\epsilon. \label{eq:aminRay}
\end{align}
Note that when $\epsilon$ decreases, $\alpha_{\min}(\epsilon)$ increases such that the range of $\alpha$ becomes smaller but at the same time the error term $O(\mathrm{SNR}^{1+\alpha+\epsilon})$ vanishes more slowly: improving the \textit{resolution} of the bandwidth occupancy range comes at the price of decreasing the \textit{accuracy} of the capacity polynomial approximation.
We can make an approximate selection of $\epsilon$ such that polynomial error term is in the order of a $p$-percent of the term $\mathrm{SNR}^{1+\alpha}$, i.e., finding  $\epsilon(p)$ such that 
\[\mathrm{SNR}^{1+\alpha}>\frac{100}{p}\mathrm{SNR}^{1+\alpha+\epsilon(p)}.\]
This generates a family of narrower estimated margins $[\alpha_{\min}(p),\alpha_{\max}]$ parametrized by the pre-selected error percentage $p\%$ by raising the lower bracket.

On the other hand, we can bound $\alpha$ using the critical bandwidth occupancy interval in Lemma~\ref{lem:critB} in combination with $\delta{=}\mathrm{SNR}^{1-\alpha}$. With $\delta B{=}(\delta B)^+$ we get
\begin{align}
\alpha^+&=\frac{\log(\frac{4(N_\mathrm{t}+N_\mathrm{r})\log\pi}{N_\mathrm{t}^2}\frac{B_cT_c}{\log(B_cT_c)})}{2\log(\mathrm{SNR}^{-1})} %\nonumber\\
% &
 =\frac{\log(\frac{(N_\mathrm{t}+N_\mathrm{r})^2}{N_\mathrm{t}^2}B_cT_c)}{2\log(\mathrm{SNR}^{-1})}-
\frac{\log(\frac{(N_\mathrm{t}+N_\mathrm{r})\log(B_cT_c)}{4\log\pi})}{2\log(\mathrm{SNR}^{-1})},
\end{align}
and with $\delta B=(\delta B)^-$ we get
\begin{align}
\alpha^- & =\frac{\log(\frac{1}{ 4(N_{\mathrm{t}}+N_{\mathrm{r}})\log\pi}\frac{B_cT_c}{\log(B_cT_c)})}{2\log(\mathrm{SNR}^{-1})}
% \nonumber\\
% &
 =\frac{\log(\frac{(N_\mathrm{t}+N_\mathrm{r})^2}{N_\mathrm{t}^2}B_cT_c)}{2\log(\mathrm{SNR}^{-1})}
-\frac{\log(4\log\pi\frac{(N_\mathrm{t}+N_\mathrm{r})^3}{N_\mathrm{t}^2}\log(B_cT_c))}{2\log(\mathrm{SNR}^{-1})}.
\end{align}
Recall that for any $\epsilon>0$ we have
$(N_{\mathrm{t}}+N_{\mathrm{r}})\log(B_cT_c)\leq {\displaystyle \lim_{\mathrm{SNR}\rightarrow0}}\mathrm{SNR}^{-2\epsilon}$. 
 This means that 
 we can show that $\alpha_{\max}>\alpha^+>\alpha_{\min}(\epsilon)$, and the interval between the three vanishes as $\epsilon{\to}0$.

\begin{remark}
All the results coincide in that $\alpha{\propto} \log(B_cT_c)$, making capacity of channels with low $B_cT_c$ approach their wideband limit very slowly with $SNR{\to}0$ and channels with high $B_cT_c$ approach the wideband limit faster. This is the main intuition of the results in \cite{Ray2007noncoherent}:  non-coherent channels approach the coherent channel capacity when coherence length is large enough.
\end{remark}

\subsection{Illustration}

\begin{figure}[t]
\centering
 \includegraphics[width=.6\columnwidth]{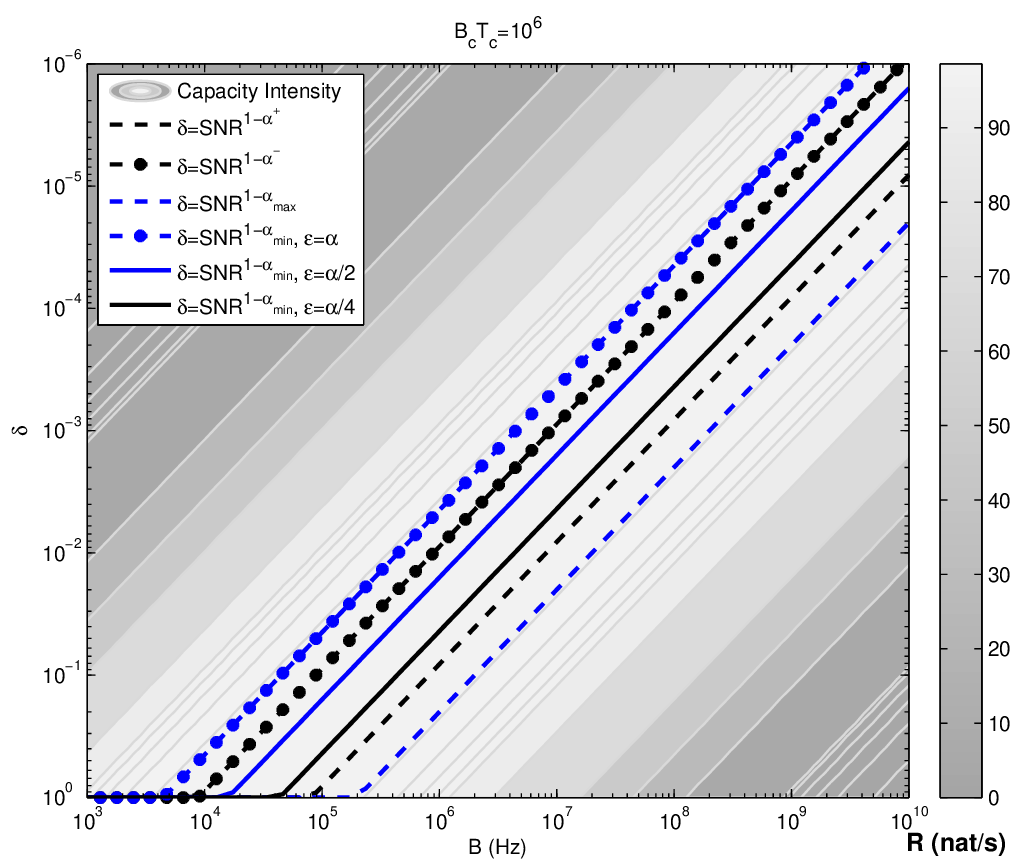}
 \includegraphics[width=.6\columnwidth]{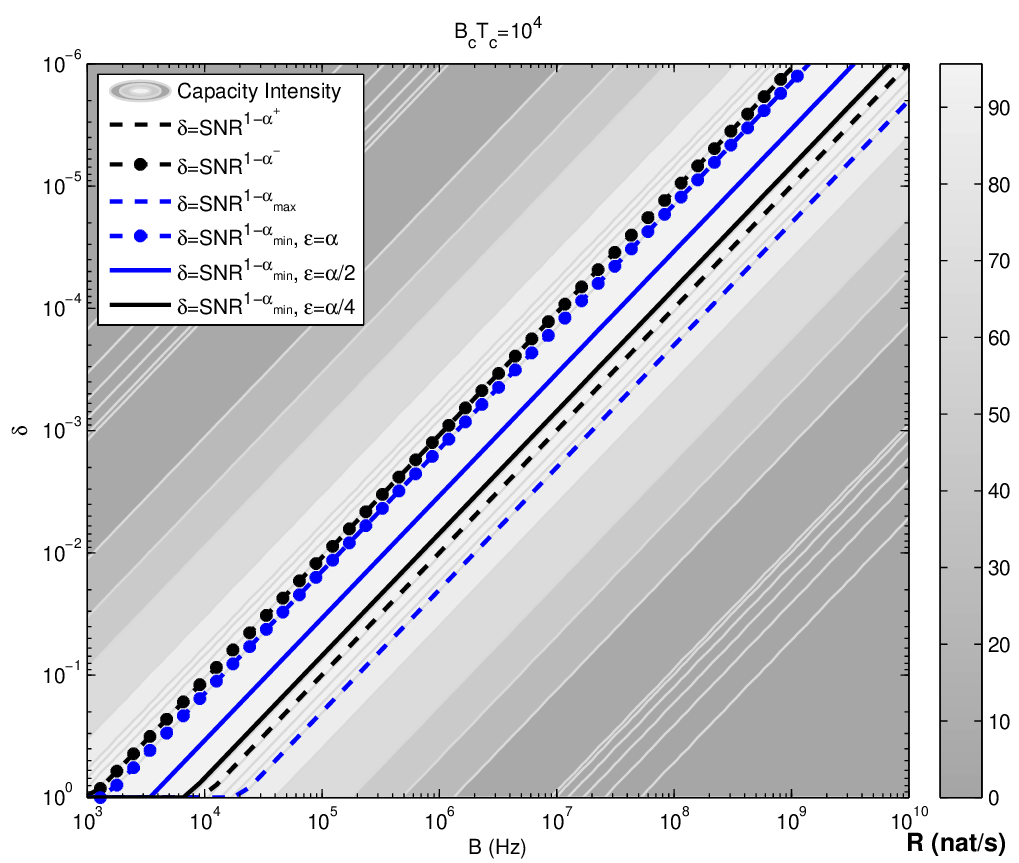}
 \caption{Capacity lower bound on the plane $(\delta,B)$ with $P/N_0{=}20$dB, $B_cT_c{=}10^6$ (first graph) and $B_cT_c{=}10^4$ (second graph). Curves with $\alpha_{\min}(\epsilon)$ are generated with $\epsilon{=}\alpha, \alpha/2, \alpha/4$, respectively.}
 \label{fig:CaphighBcTc}
\end{figure}

\begin{figure}[ht]
 \centering
 \includegraphics[width=.6\columnwidth]{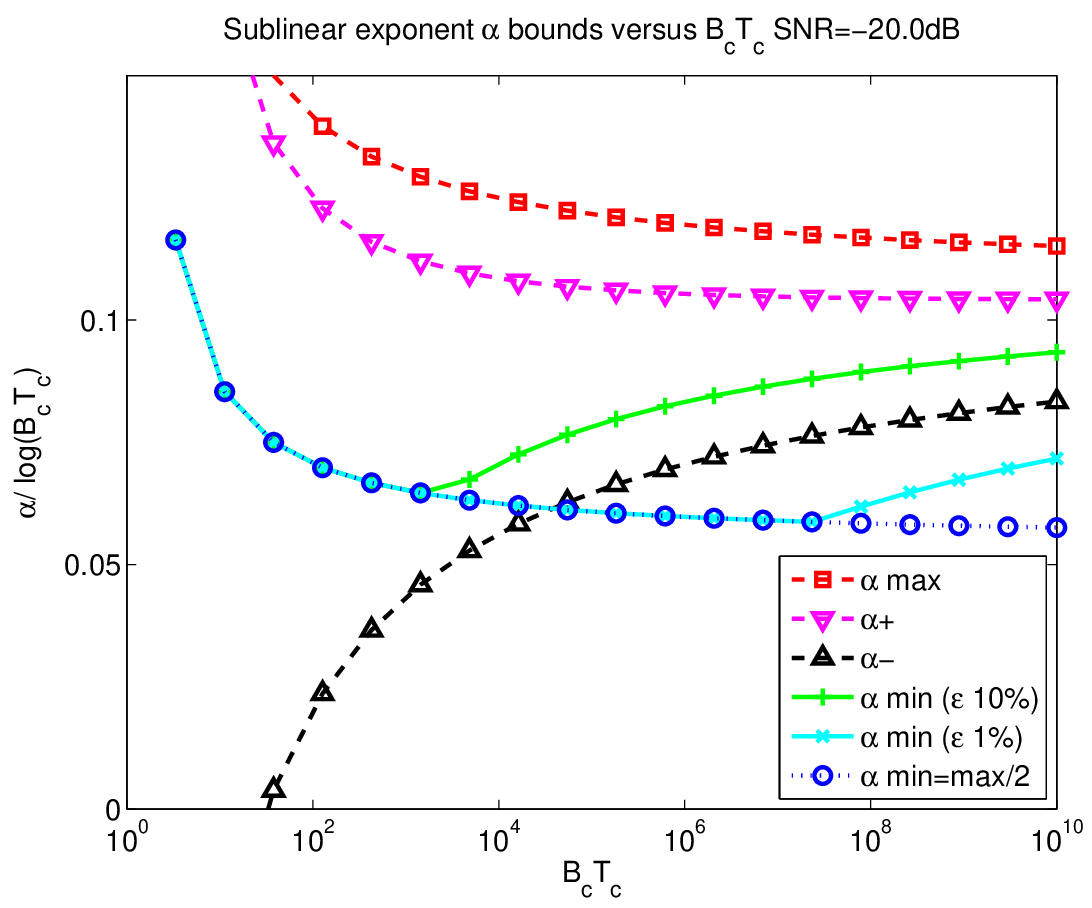}
 \caption{Evolution of $\frac{\alpha}{\log(B_cT_c)}$ versus $B_cT_c$ with $\mathrm{SNR}=-20$dB.}
 \label{fig:alphavsBcTc}
\end{figure}

We plot the capacity lower bound on the plane $(\delta, B)$ in Fig.~\ref{fig:CaphighBcTc} for $B_cT_c{=}10^6$ (first graph) and  for $B_cT_c{=}10^4$ (second graph).
The peak capacity is achievable in a region with constant product $\delta B$, starting at relatively large bandwidths, and both estimations of the optimal region are narrow.
The choice of $\epsilon$ determines the polynomial lower bound and therefore the range $[\alpha_{\min}(\epsilon),\alpha_{\max}]$. We can generate a set of estimations $\alpha_{\min}(\epsilon)$ by fine-tuning $\epsilon$ within the range $(0,\alpha)$, as shown by the curves corresponding to $\alpha_{\min}(\epsilon)$ with $\epsilon{=}\alpha/2, \alpha/4$, respectively. Note that the  conservative choice $\epsilon{=}\alpha$ leads to the widest possible range for $[\alpha_{\max}/2,\alpha_{\max}]$.
On the other hand, the resolution of the estimators from our own analysis $[\alpha^-,\alpha^+]$ depends only on the value of $B_cT_c$, and its range becomes smaller as $B_cT_c$ increases.

Since the resolution of the estimation by $[\alpha^-,\alpha^+]$ relies on $B_cT_c$ and the relative margin of $[\alpha_{\min},\alpha_{\max}]$ depends on  $\epsilon$,
we show in Fig.~\ref{fig:alphavsBcTc} the evolution of the two boundary methods with $\epsilon$ and $B_cT_c$. The method \cite{Ray2007noncoherent} produces the highest upper bound $\alpha_{\max}$  that does not change, and a family of lower bounds $\alpha_{\min}(p)$  depicted in the figure for errors of $1\%$ and $10\%$ and its lowest bound $\alpha_{\max}/2$. Note that at low coherence length, $B_cT_c$, the limit $\alpha_{\mathrm{min}}>\alpha_{\mathrm{max}}/2$ makes it impossible to select values of $\epsilon$ corresponding with a polynomial accuracy of $1\%$, and then $10\%$. This shows that the polynomial rate with peaky signaling in \cite{Ray2007noncoherent} also displays a gap from $C^\infty$ decreasing with $B_cT_c$. On the other hand, the critical bandwidth occupancy method produces boundaries that are loose at low coherence length but improve significantly when this parameter grows and that do not pay for tightness a price in accuracy of the polynomial approximation.

\section{Conclusions}\label{sec:conclusion}

In this paper we have unified the study of the rate approximations to $C^\infty$ for peaky and non-peaky signaling in non-coherent wideband fading channels where energy rather than spectrum is the limiting resource.
We have generalized the critical bandwidth analysis %for non-peaky signaling in SISO systems developed in
\cite{journals/twc/LozanoP12} to families of signaling schemes %in MIMO systems
with varying bandwidth $B$ and transmission duty-cycle $\delta{\in}(0,1]$  to allow arbitrary levels of signal peakiness. We introduce the metric of bandwidth occupancy to measure the average bandwidth usage over time and define it as $\delta B$, the product between the bandwidth and the fraction of time it is in use.
Our main result shows the existence of a fundamental limit on the bandwidth occupancy in non-coherent channels for any level of signal peakiness. For all signaling schemes with the same bandwidth occupancy, as the bandwidth occupancy approaches its critical value $(\delta B)_\mathrm{crit}$, rates converge with the same asymptotic behavior to the same almost-linear in power value (measured in nats/s)
\[C(B,\delta=\frac{(\delta B)_\mathrm{crit}}{B})\geq \frac{PN_{\mathrm{r}}}{N_0}\left[1-\sqrt{\frac{1+\log (B_cT_c)}{B_cT_c}(\kappa-2+N_{\mathrm{t}}+N_{\mathrm{r}})\log\pi}\right],\]
where $T_c$ is the coherence time and $B_c$ is the coherence bandwidth. The rates decrease to zero as the bandwidth occupancy goes to infinity. Moreover, we provide upper and lower bounds to this critical value. The bounds have the same growth with $B_cT_c$ and $\frac{P}{N_0}$, and they only differ on a constant term.

To characterize 
the relation between the capacity with a tunable peakiness constraint $C(B,\delta)$ and the unconstrained non-coherent wideband capacity $C(B)$, we
rewrite the above capacity expression as a polynomial equivalent to the analysis in \cite{Ray2007noncoherent}. We have recovered the results in \cite[Th. 1-Th. 3]{Ray2007noncoherent} and obtained the almost-linear polynomial expressions for capacity in the limit $\delta B\to (\delta B)_{\mathrm{crit}}$ with a dominant sub-linear term $\mathrm{SNR}^{\alpha}$. 
As the bandwidth occupancy approaches the limit, capacity approaches the power-limited wideband limit with a speed of convergence determined by $\mathrm{SNR}^{1+\alpha}$,  which approaches that of coherent channels as $B_cT_c\rightarrow\infty$. The fundamental nature of the bandwidth occupancy measure reflects the fact that capacity of any signaling scheme is contained within the same bounds as long as the product $\delta B$ is constant.

Within this framework, limited bandwidth transmission with non-peaky signaling and unlimited bandwidth transmission with peaky signaling, which have been treated as very different  schemes, are shown to be merely two extreme points in a continuous range of transmission strategies within the same bounds as long as they have the same amount of bandwidth occupancy.  This suggest that for the practical goal of operating at a \textit{rate very close to $C^\infty$}, all pairs $(B,\delta)$ with the optimal occupancy do not exhibit significant differences. Achieving capacity, i.e. the supremum rate, may on the other hand only be possible in some specific distributions.
The selected peakiness  $\delta=\mathrm{SNR}^{1-\alpha}$ in \cite{Ray2007noncoherent} becomes invalid if $\mathrm{SNR}{>}1$ (as $\delta\leq1$ by design), whereas our model determines peakiness through $\delta B{<}(\delta B)_{\mathrm{crit}}$, a quantity that is well defined for all values of SNR. This gives the intuition that below the critical point it would be questionable to claim that the frequency-selective channel is in the wideband regime, and therefore regular non-peaky transmissions with full bandwidth occupancy must be employed. Beyond the critical point,  both signaling schemes provide the same capacity limit.

We have shown that most of the advantage of peaky signaling stems from harnessing power for long periods of time to transmit some infrequent flashes with boosted power, without encoding information in the position of the active symbols as in ON/OFF modulations.
Moreover, this power boost does not in fact outperform non-peaky transmission with the optimal bandwidth, which means that in practical systems the amount of peakiness and the bandwidth may be chosen at will as long as the maximum occupancy level is respected.

Our analysis has some limitations. Firstly, the potential spatial correlation among MIMO antennas is not accounted for. Secondly, although our capacity lower bounds are valid for general fading channels, our upper bound and critical bandwidth occupancy expressions assume Rayleigh fading. Besides, the performance of a signaling system with practical channel estimation techniques~\cite{Jindal2010}, peak constrained signals~\cite{Hajek2002b, Zhang07, Hajek2005, Lapidoth05}, finite modulation options, and non-ideal decoders may be degraded as compared to the theoretical bounds provided in this paper.

\appendices

\section{Justification of Our Fading Model Choice}
\label{app:dvSysMod}

As a general case, a wireless channel is modeled as a set $\mathcal{L}$ of paths, where each path $\ell\in\mathcal{L}$ is defined by a group delay $\tau_\ell$, a phase of arrival $\theta_\ell$, and an impulse response  $ {h}_\ell(t)$. For a pair of antennas $(u,v)$ with received signal $r^{(v)}(t)$ and transmitted signal $s^{(u)}(t)$,  we have
\begin{align}
 r^{(v)}(t) =s^{(u)}(t)*\sum_{\ell\in\mathcal{L}} {h}_\ell^{(u,v)}(t-\tau^{(u,v)}_\ell)e^{j\theta_\ell^{(u,v)}} + z^{(v)}(t)
 %\nonumber \\
	= s^{(u)}(t)*h^{(u,v)}(t) + z^{(v)}(t),
\end{align}
where $z^{(v)}(t)$ is the AWGN noise, and the channel delay spread $D$ and coherence time $T_c$ are determined by the aggregate channel impulse response $h^{(u,v)}(t)$. Traditionally, $ {h}_\ell^{(u,v)}(t)$s are scalar gains or narrow pulses that can be approximated by the Dirac delta function, in which case the set $\mathcal{L}$ would be a sort of ``ray tracing'' of perfect reflections of the signal with a scalar gain. However, recent mmWave meassurements have found much higher delay-spread values \cite{Rappaport2013} than those predicted in ray-tracing calculations \cite{Ryu2015}. This may be due to rich scattering from small objects in mmWave fading channels, which are not so sparse in practice. This is due to the fact that, although there are few arrival direction ``clusters'', in each cluster energy arrivals spread along many angular directions \cite{mustafa2013mmWave}. Therefore each arrival direction sees the additive effect of a large number of scattered reflections, not a single path, and each $ {h}_\ell^{(u,v)}(t)$ has a delay spread, instead of a scalar channel gain. The construction of discrete-time system models falls into the following three regimes depending on the sampling rate:

\begin{itemize}
\item Sampling $h^{(u,v)}(t)$ at low rate, all the energy in the delay spread $D$ would be captured by a single sampling interval, so the resulting discrete channel would be a scalar coefficient, which is approximately Gaussian distributed due to the law of large numbers. This is called the \textbf{narrowband}, or \textbf{frequency-flat} channel.
\item Sampling at higher rate would make the energy in $D$ be captured in multiple sampling intervals, each with an independent scalar coefficient. This is called the \textbf{wideband} channel, or \textbf{frequency selective} with rich scattering environment.
\item The third regime occurs when the number of sampling bins is much larger than the number of paths in  $\mathcal{L}$. The sampled channel coefficients are sparse and not Gaussian distributed. This is called the \textbf{ultra-wideband}.
\end{itemize}

We consider the wideband fading model to be relevant in mmWave communications where rich scattering and longer delay spread was observed~\cite{Rappaport2013}.
Our discrete equivalent channel is derived from the propagation described above by employing the classic framework of a Nyquist sampling at frequency $B$, the consideration of frequency-domain signaling with a $K$-point DFT, satisfying $K=BT_c$, and a cyclic prefix of negligible duration $M=B/B_c=K/B_cT_c\ll K$. The same channel model is employed in \cite{journals/twc/LozanoP12}.

\section{Compatibility With Another Common Model}
\label{app:bkSysMod}
In \cite{Ray2007noncoherent} the signals are divided into a set of $M=B/B_c$ narrowband channels (a.k.a. frequency bins) with encoding symbols defined with a symbol period of $1/B_c$. Each narrowband channel can be perfectly sampled at a rate of just $1$ sample per symbol period, and there are $M$ parallel frequency bands that produce $M$ samples per symbol period. In this scheme multiple symbols see the same channel realization and the channel coherence length is a block of $L_c=T_cB_c$ consecutive symbols. By indexing with $m$ the independent frequency bins and with $\ell$ the consecutive periods on the same channel block realization, we get the model
\begin{equation}
 \y[m,\ell]=\Hb[m,\ell]\x[m,\ell]+\z[m,\ell],
\end{equation}
where $\Hb[m,\ell]$ remains unchanged for $\ell=1,\ldots,L_c$.
To exploit channel coherence, the encoding process must design the transmitted signal for the $L_c$ consecutive symbols jointly, and the encoding model is represented with matrices as
\begin{equation}\label{eq:FBmodel}
 \Y[m]=\Hb[m]\X[m]+\Z[m],
\end{equation}
where the dimensions are $N_\mathrm{r}\times L_c=(N_\mathrm{r}\times N_\mathrm{t})(N_\mathrm{t}\times L_c)$.

In this model, for every encoding interval of length $L_c$ and across all $M$ frequency bins there are a total of $ML_c{=}K$ complex valued coefficients.
Therefore, this channel model provides exactly the same number of signaling dimensions for transmission as the model we have derived.
But the representations of the channel variation are different.  In this model there are fewer channel coefficients, each of them is i.i.d. and  identically repeated for every $L_c$ consecutive symbols. Whereas our derived model \eqref{eq:sysModel} supports any type of channel correlation, not only repetition, as long as there are $K$ correlated  coefficients generated by a fraction $1/B_cT_c$ of independent random variables. It is possible to represent the system model \eqref{eq:FBmodel} with repeated identical channel coefficients in our derived model format by replacing the matrix notation $\Hb[m]\X[m]$ with our vectorized notation $\Hb\x$ where $\Hb$ is a block-diagonal matrix with the values of $\Hb[m]$ in its main diagonal and zeros in the upper and lower triangles as in \eqref{eq:MatrixChannel}.

\section{Equivalence in Signaling Representation}
\label{app:bkSigRepr}

Our channel model uses Nyquist sampling at the full $B$ and therefore it is able to represent any signal with this bandwidth without loss. For the sake of completeness we will propose the exact formulation to implement a valid signal in the model of \cite{Ray2007noncoherent} (hereafter, filter-bank model) with our model (hereafter, OFDM model) using only preprocessing linear matrices. With this we show that any signal possible in the filter-bank model can be transmitted through the OFDM model, and therefore capacity results in our model are fully compatible.

Without loss of generality, let us assume a SISO channel and unit power to simplify notation. Assume also that the integers $K{=}\lceil T_cB\rceil$, $M{=}\lceil B/B_c\rceil$ and $L_c{=}\lceil T_cB_c\rceil$ are satisfied exactly so we may use simply $K{=}ML_c$.
In continuous time, the filter-bank model is represented in %\eqref{eq:ctsignal},
%%
%\begin{figure*}[!t]
%\setcounter{MYtempeqncnt}{\value{equation}}
%\setcounter{equation}{32}
%\normalsize
\begin{equation}
r(t){=}\sum_{i=-\infty}^\infty \delta(t{-}iT_c)*\left(\sum_{\ell=0}^{L_c-1}\delta(t{-}\ell/B_c)*\left(\sum_{m=0}^{M-1} h_i[m,\ell]x_i[m,\ell]\mathrm{sinc}(tB_c)e^{-j2\pi mB_c t}\right)\right){+}z(t),
%\nonumber
\label{eq:ctsignal}
\end{equation}
%\hrulefill
%\vspace*{2pt}
%\setcounter{equation}{\value{MYtempeqncnt}+1}
%\end{figure*}
%
where 
multiplication by $h_i[m,\ell]$ and $x_i[m,\ell]$ assigns the scalar value received in each frequency bin $m\in\{0,\ldots,M{-}1\}$ and in each transmit symbol period $\ell\in\{0,\ldots,L_c{-}1\}$.

We separate the encoding for each channel block realization indexed by $i$, drop the index, and use the fact that the channel coefficient in each frequency bin remains the same for all symbols to take away the index $\ell$ from $h[m]$. This gives
\begin{align} \label{eq:ct_i}
y(t)=
&\sum_{\ell=0}^{L_c-1}\delta(t{-}\ell/B_c)*\left(\sum_{m=0}^{M-1} h[m]x[m,\ell]\mathrm{sinc}(tB_c)e^{j2\pi mB_c t}\right){+}z(t).
\end{align}
With this continuous-time signal, we apply Nyquist sampling at rate $B$ to generate  $K{=}BT_c$ samples per sequence.
Notice that for integer $M{=}B/B_c$, the discrete sinc function is $\mathrm{sinc}[n/M]\triangleq\mathrm{sinc}(\frac{nB_c}{B})$ and the delta delay on index $n$ is $\ell M$. We can represent \eqref{eq:ct_i} by
\begin{align}
y[n]{=} \sum_{\ell=0}^{L_c-1}\delta[n{-}\ell M]*\left(\sum_{m=0}^{M-1} h[m]x[m,\ell]\mathrm{sinc}[\frac{n}{M}]e^{j2\pi \frac{n}{K}mL_c}\right){+}z[n], \ n=0,\ldots,K{-}1.  %\nonumber
% \hspace{-1mm} y[n]{=} \ssum_{\ell}\delta[n{-}\ell M]*(\ssum_{m}h[m]x[m,\ell]\mathrm{sinc}[\frac{n}{M}]e^{j2\pi \frac{n}{K}mL_c}){+}z[n]. \nonumber
\end{align}
We compute the $K$-point DFT, 
\begin{align}
y[k]
%& = \sum_{\ell, m}    e^{-j2\pi \frac{k}{K}\ell M}h[m]x[m,\ell]\mathrm{rec}[\frac{kM}{K}]*\delta[k-mL_c]{+}z[k],\nonumber\\
& = \sum_{\ell, m}    e^{-j2\pi \frac{k\ell}{L_c}}h[m]x[m,\ell]\mathrm{rec}[k/L_c-m]{+}z[k].
\end{align}

The rectangular window equals one only when $\lfloor k/L_c\rfloor{=}m$. 
By representing $k{<}K$ as $k=u*L_c + v$ with $u\triangleq\lfloor k/L_c\rfloor$ and $v\triangleq k \mod L_c$, we obtain 
\begin{equation}\nonumber
y[k]=h[u]\sum_{\ell=0}^{L_c-1} e^{j2\pi \frac{v\ell}{L_c}}x[u,\ell]+z[k], \ \mbox{with }
\left\{\begin{array}{l}
u{=}\lfloor k/L_c\rfloor,\\
v{=}k \mod L_c.
\end{array}\right.
\end{equation}
Now we can see that the sum is actually the $v$th element in the $L_c$-IDFT of the sequence $x[u,\ell]$.
Since the IDFT of a sequence $\mathbf{a}{=}(a_1\dots a_{L_c})^T$ can be written as a matrix product $\mathrm{IDFT}(\mathbf{a})=\F\mathbf{a}$,
We can represent the system model the same way as our matrix channel notation as
\begin{equation}
 \y=\Hb\Phi\x+\z,
\end{equation}
where $\Hb$ for SISO is a $K\times K$ diagonal matrix with its $k$-th diagonal element $h[u]$,  $\x$ is $K\times 1$ with $\x^{(k)}=x[u,v]$. The $L_c$-IDFT is computed by the block-diagonal square  matrix
\begin{equation}
\label{eq:MatrixlDFT}
 \Phi=\left(\begin{array}{c|c|c}
           \F&\dots&\zero\\\hline
           \vdots&\ddots&\vdots\\\hline
           \zero&\dots&\F\\
          \end{array}\right).
\end{equation}

This shows that any channel of the filter-bank model can be represented by the OFDM model using a channel matrix $\tilde{\Hb}=\Hb\Phi$. The reciprocal compatibility can be proven by taking a precoding DFT matrix at the transmitter, $\tilde{\x}=\Phi^{\dag}\x$, which leads to 
\begin{align}\nonumber
\y=\tilde{\Hb}\tilde{\x}=\Hb\Phi\tilde{\x}=\Hb\Phi\Phi^{\dag}\x=\Hb\x.
\end{align} 
The multiplication by $\Phi^{\dag}$ is unitary, so if the OFDM model uses $\x\sim \mathcal{CN}(\vmu,\Sb)$ and $\Phi$ is a full-$K$-rank square orthonormal matrix, then $\Phi^{\dag}\x\sim \mathcal{CN}(\Phi^{\dag}\vmu,\Phi^{\dag}\Sb\Phi)$. Gaussian distribution is maintained when the channel model is changed, and the mutual information results for both channel models supported by our bounds based on Gaussian inputs are completely equivalent.

\section{Proof of Lemma 1}
\label{app:Prooflem1}

Since the receiver knows which phase the duty cycle is in (e.g., scheduled according to a pseudo-random sequence),  {}{the rate can be determined via} the  chain rule
 \begin{equation}
 \label{eq:chainrule}
  {\frac{1}{T_c}\CInf{\x}{\y}{c}=\frac{\delta}{T_c} \CInf{\x}{\y}{c{=}1} + (1{-}\delta)\cdot 0 =\frac{\delta}{T_c}\CInf{\x,\Hb}{\y}{c{=}1} -\frac{\delta}{T_c}\CInf{\Hb}{\y}{\x,c{=}1}},
 \end{equation}
 {}{where the first step comes from the fact that $\Hb\x{=}0$ in the idle block ($c{=}0$) and $P_r(c{=}1)=\delta$.
 During the active block} the input follows a Gaussian distribution $\mathcal{CN}(0,\frac{P}{\delta B N_0})$ {}{and the first term in \eqref{eq:chainrule} can be lower bounded by} 
 \begin{align}
  \frac{{}{\delta}}{T_c}\CInf{\x,\Hb}{\y}{{}{c{=}1}}&\geq\frac{{}{\delta}}{T_c}\CInf{\x}{\y}{\Hb{}{,c{=}1}}
	       =\delta\times \Ex{\Hb}{\frac{1}{T_c}\log\det(\I_{KN_\mathrm{r}}{+}\frac{P}{\delta BN_{\mathrm{t}}N_0}\Hb\Hb^{\dag})},
 \end{align}
where the first step is from the non-negativity of mutual information, and the second is due to independence of channel coefficient $\Hb$ in each subcarrier and transmit antenna.
Furthermore,
 \begin{align}
 &\delta \textnormal{E}_\Hb \bigg[\frac{1}{T_c}\log\det(\I_{KN_\mathrm{r}}+\frac{P}{\delta BN_{\mathrm{t}}N_0}\Hb\Hb^{\dag})\bigg] % \nonumber\\
  {=} \delta \frac{K}{T_c} \textnormal{E}_\Hb \bigg[\log\det(\I_{N_\mathrm{r}}+\frac{P}{\delta BN_{\mathrm{t}}N_0} \hat{\Hb}\hat{\Hb}^{\dag})\bigg]\nonumber\\
	&{\stackrel{(a)}{=}} \delta \frac{K}{T_c}\sum_{i=1}^{\min(N_{\mathrm{t}},N_{\mathrm{r}})}\Ex{\Hb}{\log(1+\frac{P}{\delta BN_{\mathrm{t}}N_0}\lambda_i)} \nonumber\\
  %&\delta B\sum_{i=1}^{\min(N_{\mathrm{t}},N_{\mathrm{r}})}\Ex{\Hb}{\log\det(1+\frac{P}{\delta BN_{\mathrm{t}}N_0}\lambda_i)} \nonumber\\
  &{\stackrel{(b)}{\geq}}  \delta B\Ex{\Hb}{\frac{P \ \tr(\hat{\Hb}\hat{\Hb}^{\dag})}{\delta BN_{\mathrm{t}}N_0}-\left(\frac{P}{\delta BN_{\mathrm{t}}N_0}\right)^2\frac{\tr\left((\hat{\Hb}\hat{\Hb}^{\dag})^2\right)}{2}} \nonumber\\
%     &{\stackrel{(c)}{=}}  \frac{PN_{\mathrm{r}}}{N_0}\left[1-\frac{P/(\delta B)}{2N_{\mathrm{r}}N_{\mathrm{t}}^2N_0}\Ex{\Hb}{\sum_{t,u=1}^{N_{\mathrm{r}}}\sum_{r,v=1}^{N_{\mathrm{t}}}h_{t,r}h_{u,r}^*h_{t,v}^*h_{u,v}}\right] \nonumber\\
    &{\stackrel{(c)}{=}} \frac{PN_{\mathrm{r}}}{N_0}\left[1{-}\frac{P/(\delta B)}{2N_{\mathrm{r}}N_{\mathrm{t}}^2N_0}\textnormal{E}_H \bigg[ \sum_{t,r}|h_{t,r}|^4{+} \sum_{t\neq u,r}|h_{t,r}|^2|h_{u,r}|^2 
%     \right.  \nonumber\\  & \hspace{45mm} \left.
    {+} \sum_{t,r\neq v}|h_{t,r}|^2|h_{t,v}|^2 {+} \sum_{t\neq u, r\neq v}h_{t,r}h_{u,r}^*h_{t,v}^*h_{u,v}\bigg] \right] \nonumber\\
    &{\stackrel{(d)}{=}} \frac{PN_{\mathrm{r}}}{N_0}\left[1-\frac{P\left( N_{\mathrm{t}}N_{\mathrm{r}}\kappa+N_{\mathrm{t}}N_{\mathrm{r}}(N_{\mathrm{r}}{-}1)+N_{\mathrm{r}}N_{\mathrm{t}}(N_{\mathrm{t}}{-}1) \right)}{2\delta BN_{\mathrm{r}}N_{\mathrm{t}}^2N_0}\right] \nonumber\\
    & {=}  \frac{PN_{\mathrm{r}}}{N_0}\left[1-\frac{P}{2\delta BN_{\mathrm{t}}N_0}\left(\kappa-2+N_{\mathrm{t}}+N_{\mathrm{r}} \right)\right],
  \end{align}
	where $\lambda_i$ are eigenvalues of $\hat{\Hb}\hat{\Hb}^{\dag}$ with $\hat{\Hb} {=} [h_{t,r}]_{N_{\mathrm{r}}\times N_{\mathrm{t}}}$ representing the  diagonal blocks of $\Hb$, and $h_{t,r}$ is the $(r,t)$-th element in $\hat{\Hb}$. Equation $(a)$ comes from the fact that $\Hb[k]$ are identically distributed for all $k{=}0,\ldots,K{-}1$, $(b)$ is due to  $\log(1{+}x){\geq} x{-}x^2/2$ for $x{\to}0$ and the fact that $\sum_i \lambda_i{=}\tr(\hat{\Hb}\hat{\Hb}^{\dag})$ and $\sum_i \lambda^2_i{=}\tr((\hat{\Hb}\hat{\Hb}^{\dag})^2)$. Equation $(c)$ is 
% 	due to $\Ex{H}{|h|^2}{=}1$ and $(d)$ is 
	by careful rearrangement. Equation $(d)$ comes from $\Ex{\Hb}{|h|^2}{=}1$, $\Ex{\Hb}{|h|^4}{=}\kappa$, $\Ex{\Hb}{h}{=}0$, and independence of matrix entries.

To upper bound the second term we choose $\Hb$ to be Rayleigh fading (with the maximum entropy) and interpret $\mathbf{x}$ as a pilot signal that gives side information between $\Hb$ and $\y$. 
\begin{align}\label{eq:HestUpp}
\CInf{\Hb}{\y}{\x {}{,c{=}1} }  \leq  \CInf{\Hb_{\mathrm{Gaussian}}}{\y}{\x_\textnormal{Pilots Signal}{}{,c{=}1}} \triangleq \overline{\CInf{\Hb}{\y}{\x{}{,c{=}1}}}.
\end{align}
An example for channel estimation would be a system where the pilot signal transmitted on antenna $u$ is a $uM$ times delayed version of the signal on antenna $1$. After transmitting $K$ pilot symbols, at each receive antenna a $K$-equation $MN_{\mathrm{t}}$-unknowns linear  estimation problem  is established and can be solved using the MMSE estimator.

Let $\Lambda^{(v)}$ be the $ MN_{\mathrm{t}}\times MN_{\mathrm{t}}$ diagonal matrix containing in its $uM{+}m$  diagonal element $g_{uM+m}=\Ex{}{|h[m]^{(u,v)}|^2}$ (the gain of the $m$-th channel tap in the $(u,v)$ transmit and receive antenna pair), and let $\Xi$ be a $K\times MN_{\mathrm{t}}$ circulant matrix ($MN_{\mathrm{t}}{<}K$) containing $\tilde{\x}_{(i-j)\mod K}$ in its $(i,j)$-th coefficient, where $\tilde{\x}{=}\x/\sqrt{P}$ is unit-power pilot signal. Notice that the mention of pilot signals here is to upper bound a mutual information term, rather than implementing a practical channel estimation as required in a coherent receiver.
 Exploiting the fact that channel estimation is carried out on each receive antenna concurrently based on the hypothetical pilot signal $\Xi$ from all transmit antennas,
 we get that the upper bound results in
\begin{align}
\frac{{}{\delta}}{T_c}\overline{\CInf{\Hb}{\y}{\x{}{,c{=}1}}}& {=} \frac{\delta}{T_c}\sum_{v=1}^{N_\mathrm{r}}\Ex{}{\log\det\left(\I+\frac{P/(\delta B)}{N_{\mathrm{t}}N_0}\Xi^{\dag}\Xi\Lambda^{(v)}\right)} \nonumber\\
&{\stackrel{(a)}{\leq}} \frac{\delta N_\mathrm{r}}{T_c}MN_{\mathrm{t}}\Ex{}{\log\left(\frac{1}{MN_{\mathrm{t}}}\tr\left(\I+\frac{P/(\delta B)}{N_{\mathrm{t}}N_0}\Xi^{\dag}\Xi\Lambda^{(1)}\right)\right)}\nonumber\\
&{\stackrel{(b)}{=}}  \frac{\delta B N_\mathrm{r}N_{\mathrm{t}}}{B_cT_c} \Ex{}{\log\hspace{-1mm}\left(\hspace{-1mm}1{+}\frac{P/(\delta B)}{M N_{\mathrm{t}}^2N_0}\sum_{n=1}^{MN_{\mathrm{t}}}  g_{n} \sum_{k=0}^{K-1}|\tilde{x}[k{-}n{-}1]|^2\right)} \nonumber\\
&{\stackrel{(c)}{\leq}} \frac{\delta B N_\mathrm{r}N_{\mathrm{t}}}{B_cT_c} \log\left(1{+}\frac{P/(\delta B)}{M N_{\mathrm{t}}^2N_0}K \sum_{n=1}^{MN_{\mathrm{t}}} g_{n}\Ex{}{\frac{1}{K}\sum_{k=0}^{K-1}|\tilde{x}[k]|^2}\right) \nonumber\\
&{\stackrel{(d)}{=}} \frac{\delta B N_\mathrm{r}N_{\mathrm{t}}}{B_cT_c} \log\left(1+\frac{P/(\delta B)}{ M N_{\mathrm{t}}^2N_0}K\sum_{n=1}^{MN_{\mathrm{t}}} g_{n}\right) \nonumber\\
&{\stackrel{(e)}{\leq}} \frac{\delta B N_\mathrm{r}N_{\mathrm{t}}}{B_cT_c} \log\left(1+\frac{P/(\delta B)}{M N_{\mathrm{t}} N_0}K\right) \nonumber\\	
&{\stackrel{(f)}{=}} \frac{\delta B N_\mathrm{r}N_{\mathrm{t}}}{B_cT_c} \log\left(1+\frac{P}{\delta B N_0 N_{\mathrm{t}}}(B_cT_c)\right), \label{eq:RpilotLB}		
\end{align}
where $(a)$ stems from  the AM–GM inequality and that channel gains between all antenna pairs are i.i.d, $(b)$ is due to the fact that $\Xi$ is a  circulant matrix, which has the same coefficients shifted across all its columns, so its eigenvalues are the DFT coefficients of the columns, $(c)$ is Jensen's inequality, $(d)$ derives from the fact that $\tilde{x}$ has unit power, and $(e)$ is due to the upper bound  of squared channel coefficients $\sum_{n=1}^{MN_{\mathrm{t}}} g_{n}\leq N_{\mathrm{t}}$, and $(f)$ uses $\frac{K}{M}=B_cT_c$.

\section{Proof of Lemma 2}
\label{app:Prooflem2}

Taking partial derivative of \eqref{eq:RLB} w.r.t. the product $\delta B$, we obtain
%\eqref{eq:dRLB}, shown on the top of the next page.
%\begin{figure*}[!t]
%\setcounter{MYtempeqncnt}{\value{equation}}
%\setcounter{equation}{42}
\begin{equation}
%\normalsize
\frac{\partial R^{LB}( \delta B)}{\partial (\delta B)}{=}\frac{PN_{\mathrm{r}}}{N_0}\bigg[\frac{P(\kappa{-}2{+}N_{\mathrm{t}}{+}N_{\mathrm{r}})}{2(\delta B)^2N_{\mathrm{t}}N_0}{-}\frac{N_0N_{\mathrm{t}}}{PB_cT_c}\log\left(1{+}\frac{P B_cT_c}{(\delta B)N_{\mathrm{t}}N_0}\right){+}\frac{1}{\delta B\left(1{+}\frac{P B_cT_c}{N_0N_{\mathrm{t}}(\delta B)}\right)}\bigg].
%\nonumber
\label{eq:dRLB}
\end{equation}
%\hrulefill
%\vspace*{2pt}
%\setcounter{equation}{\value{MYtempeqncnt}+1}
%\end{figure*}
Near the maximum of $R^{LB}(\delta B)$ the term $\frac{P}{(\delta B)N_0}B_cT_c$ is either ${\gg}1$ or ${\simeq} 1$; because  $R^{LB}(\delta B)$ is already approaching zero if $\frac{P}{(\delta B)N_0}B_cT_c\ll1$. This means we can make the approximation
\begin{equation}
\frac{(\kappa-2+N_{\mathrm{t}}+N_{\mathrm{r}})}{2}\simeq\frac{\log(1+\frac{P}{(\delta B)^*N_0}B_cT_c)}{\left( {P}/{(N_{\mathrm{t}}N_0(\delta B)^*)}\right)^2B_cT_c},
\end{equation}
which solves as \eqref{eq:optimB}. Evaluating $R^{LB}(\delta B^*)$ and using the same inequality in \cite{journals/twc/LozanoP12} produces  \eqref{eq:optimR}.

\section{Proof of Lemma 3}
\label{app:Prooflem3}

 We upper bound the first term in \eqref{eq:chainrule} enforcing signal bandwidth $B$ and duty cycle $\delta$.
 \begin{align}
  \frac{{}{\delta}}{T_c} \CInf{\x,\Hb}{\y}{{}{c{=}1}}&\stackrel{(a)}{=}\frac{{}{\delta}}{T_c} \dEnt{\y|{}{c{=}1}
  } {-}\frac{{}{\delta}}{T_c}\dEnt{\y|\Hb,\x{}{,c{=}1}}\nonumber\\
  &\stackrel{(b)}{=}\frac{{}{\delta}}{T_c} \dEnt{\Hb\x+\z|{}{c{=}1}} {-}\frac{{}{\delta}}{T_c}\dEnt{\z}\nonumber\\
  &\stackrel{(c)}{\leq}\frac{\delta}{T_c} \dEnt{\mathcal{CN}(0,\I \frac{P}{\delta}+BN_0)}{-}\frac{\delta}{T_c}\dEnt{\z} \label{eq:Lem3-1st}\\
	      &=\delta N_{\mathrm{r}}B\log(1+\frac{P}{\delta BN_0}),\nonumber
  \end{align}
  where (a) is from the definition of mutual information; (b) is from the channel model; (c) comes from the fact that ${\z}$ is independent of $\x$ and $\Hb$, and $ \dEnt{\Hb\x+\z|{}{c{=}1}}$ is maximized by a Gaussian distribution under the power  constraint $\frac{P}{\delta}+BN_0$.
 Use the approximation $\log(1{+}x)=x{-}x^2/2+o(x^2)$, we can rewrite \eqref{eq:Lem3-1st} as
 \begin{equation}
 \begin{split}
   \frac{{}{\delta}}{T_c} \CInf{\x,\Hb}{\y}{{}{c{=}1}}&\leq\frac{PN_{\mathrm{r}}}{N_0}\left[1-\frac{P}{2\delta BN_0}\right]+o(\frac{1}{\delta B}).
  \end{split}
 \end{equation}

For the second term of \eqref{eq:chainrule}, with the Rayleigh fading assumption, the inequality in \eqref{eq:HestUpp} is met with equality. From there on, upper bounds are found by taking a couple of minimums in the argument of the logarithm.
\begin{align}
 \frac{{}{\delta}}{T_c}\underline{\CInf{\Hb}{\y}{\x{}{,c{=}1}}}=&\frac{\delta}{T_c}\sum_{v=1}^{N_\mathrm{r}}\Ex{}{\log\det\left(\I+\frac{P/(\delta B)}{ N_{\mathrm{t}} N_0}\Xi^{\dag}\Xi\Lambda^{(v)}\right)} \nonumber\\
\stackrel{(a)}{\geq}& \frac{\delta N_\mathrm{r}}{T_c}\Ex{}{\log\det\left(\I+\frac{P g_{\min}}{\delta BN_{\mathrm{t}} N_0}\Xi^{\dag}\Xi\right)} \nonumber\\
 \stackrel{(b)}{=}& \sum_{n=1}^{MN_{\mathrm{t}}} \frac{\delta N_\mathrm{r}}{T_c}\Ex{}{\log\left(1+\frac{P g_{\min}}{\delta BN_{\mathrm{t}} N_0}\lambda_n(\Xi^{\dag}\Xi)\right)} \nonumber\\
\stackrel{(c)}{\geq}&	 \frac{ \delta BN_\mathrm{r}N_{\mathrm{t}}}{B_cT_c}\Ex{}{\log\left(1{+}\frac{P g_{\min}\psi}{\delta N_{\mathrm{t}} N_0B}B_cT_c\right)}, \label{eq:RpilotUB}
\end{align}
Equation $(a)$ is due to $g_{\min}{=}\min_{m,u,v} \Ex{}{|h[m]^{(u,v)}|^2}$ is the minimum element in the diagonals of $\Lambda^{(v)}$ and among all $v$'s, and $(b)$ stems from the relation between determinant and eigenvalues. 
Since $\Xi$ is a $K{\times}MN_t$ circulant matrix containing the power normalized vector $\x/\sqrt{P}$ in its first column, the  $n$-th eigenvalue of $\Xi^{\dag}\Xi$ is given by
\[\lambda_n(\Xi^{\dag}\Xi)=|\textstyle{}\sum_{k=0}^{K-1} \frac{\x_k}{\sqrt{P}} e^{-j2\pi \frac{kn}{MN_{\mathrm{t}}}}|^2 \triangleq  K \psi_{K,n}, \ n=1,\ldots,MN_t.\] 
Since $\Ex{}{\psi_{K,n}}\leq \frac{1}{KP} |\sum_{k=0}^{K-1}\x_k|^2\leq 1$ owing to the power constraint $\Ex{}{\x}\leq P$, we obtain $(c)$ by the fact that $B_cT_c<K$ and by  the definition of $\psi$ in (12). 
Moreover, we have $\psi>0$ because the rate penalty of non-peaky inputs in active cycles is non-zero ($\delta \CInf{\Hb}{\y}{\tilde{\x}/\sqrt{\delta}}>0$).

%\[\psi{=}\arg\min_{\psi_{K,n}} \Ex{}{\log\left(1{+}\frac{P g_{\min}\psi_{K,n}}{\delta N_{\mathrm{t}} N_0B}B_cT_c\right)},\ \mbox{ where } \textstyle{}\psi_{K,n}{\triangleq}\frac{1}{K}|\sum_{k=0}^{K-1} \frac{\x_k}{\sqrt{P}}e^{-j2\pi \frac{kn}{MN_{\mathrm{t}}}}|^2.\]
\section{Proof of Lemma 4}
\label{app:Prooflem4}

 We define $(\delta B)^\pm$ such that
\begin{equation}
\label{eq:defBpm}
 \frac{P}{(\delta B)^\pm N_0}=\sqrt{\Omega \frac{\log (B_cT_c)}{B_cT_c}}+o(\sqrt{\frac{\log (B_cT_c)}{B_cT_c}}).
\end{equation}
Substituting \eqref{eq:defBpm} into \eqref{eq:RUB} we obtain that
\begin{align}
R^{UB}(\delta B)= \frac{PN_{\mathrm{r}}}{N_0}\Bigg[1{-}\frac{1}{2}\sqrt{\Omega \frac{\log (B_cT_c)}{B_cT_c}}
{-} N_{\mathrm{t}}\frac{\Ex{ }{\log(1{+}\sqrt{\Omega B_cT_c\log (B_cT_c)} g_{\min} \psi/N_{\mathrm{t}})}}{\sqrt{\Omega B_cT_c\log (B_cT_c)}}\Bigg] 
 &\nonumber\\
+  o(\sqrt{\frac{\log (B_cT_c)}{B_cT_c}}). &
\end{align}
We separate the logarithm in two parts
\begin{align}
R^{UB}(\delta B)=\frac{PN_{\mathrm{r}}}{N_0}\Bigg[1{-}&\frac{1}{2}\sqrt{\Omega \frac{\log (B_cT_c)}{B_cT_c}}
   {-}\frac{1}{2}\frac{N_{\mathrm{t}}\log(B_cT_c)}{\sqrt{\Omega B_cT_c\log (B_cT_c)}} \nonumber \\
- & N_{\mathrm{t}}\frac{\Ex{ }{\log(\frac{1}{\sqrt{B_cT_c}}{+}\sqrt{\Omega \log (B_cT_c)} g_{\min} \psi/N_{\mathrm{t}})}}{\sqrt{\Omega B_cT_c\log (B_cT_c)}}\Bigg]
+   o(\sqrt{\frac{\log (B_cT_c)}{B_cT_c}}).
\end{align}
Since $E[\psi]\leq1$, the third negative part is also $o(\sqrt{\frac{\log (B_cT_c)}{B_cT_c}})$. We have
\begin{align}
R^{UB}(\delta B) =   \frac{PN_{\mathrm{r}}}{N_0}\left[1{-}\sqrt{\frac{\log (B_cT_c)}{B_cT_c}}\frac{1}{2}\left(\sqrt{\Omega} {+}\frac{N_{\mathrm{t}}}{\sqrt{\Omega}}\right)\right]
 + o(\sqrt{\frac{\log (B_cT_c)}{B_cT_c}}).
\end{align}

We will make this upper bound equal the achievable value in \eqref{eq:optimR},  which leads to
\begin{equation}
 \frac{1}{2}\left(\sqrt{\Omega} +\frac{N_{\mathrm{t}}}{\sqrt{\Omega}}\right)=\sqrt{(\kappa-2+N_{\mathrm{t}}+N_{\mathrm{r}})\log\pi}+o(\sqrt{\frac{B_cT_c}{\log (B_cT_c)}}).
\end{equation}
By making change of variable $\Upsilon=\Omega/N_{\mathrm{t}}$  we get
\begin{equation}
\label{eq:simmetry}
 \left(\sqrt{\Upsilon} +\frac{1}{\sqrt{\Upsilon}}\right)=2\sqrt{(\frac{\kappa-2+N_{\mathrm{r}}}{N_{\mathrm{t}}}+1)\log\pi}+o(\sqrt{\frac{B_cT_c}{\log (B_cT_c)}}).
\end{equation}
With $\kappa=2$ for Rayleigh fading, $(\frac{\kappa-2+N_{\mathrm{r}}}{N_{\mathrm{t}}}+1)\geq 1$. We obtain the following two roots of \eqref{eq:simmetry}
\begin{equation}\label{eq:Gam:true}
\begin{split}
\sqrt{\Upsilon}^-&=\sqrt{(\frac{N_{\mathrm{r}}}{N_{\mathrm{t}}}+1)\log\pi}+\sqrt{(\frac{N_{\mathrm{r}}}{N_{\mathrm{t}}}+1)\log\pi-1}+o(\sqrt{\frac{B_cT_c}{\log (B_cT_c)}}),\\
\sqrt{\Upsilon}^+&=\sqrt{(\frac{N_{\mathrm{r}}}{N_{\mathrm{t}}}+1)\log\pi}-\sqrt{(\frac{N_{\mathrm{r}}}{N_{\mathrm{t}}}+1)\log\pi-1}+o(\sqrt{\frac{B_cT_c}{\log (B_cT_c)}}).\\
\end{split}
\end{equation}
It is ready to see that
\begin{equation}\label{eq:Gam:app}
\begin{split}
\sqrt{\Omega}^- & =\sqrt{N_{\mathrm{t}}} \sqrt{\Upsilon}^- \leq 2\sqrt{(N_{\mathrm{r}}+N_{\mathrm{t}})\log\pi}+o(\textstyle{}\sqrt{\frac{B_cT_c}{\log (B_cT_c)}}), \\
\sqrt{\Omega}^+ & =\sqrt{N_{\mathrm{t}}}\sqrt{\Upsilon}^+ \geq \frac{N_{\mathrm{t}}}{2\sqrt{(N_{\mathrm{r}}+N_{\mathrm{t}})\log\pi}} +o(\textstyle{}\sqrt{\frac{B_cT_c}{\log (B_cT_c)}}).
\end{split}
\end{equation}
 Substituting them back in \eqref{eq:defBpm} we get the points $(\delta B)^-$ and $(\delta B)^+$ as shown in \eqref{eq:dBcritical}. Therefore the true achievement of the maximum can only occur in the range $(\delta B)_\mathrm{crit}\in[(\delta B)^-,(\delta B)^+]$.

\section*{Acknowledgment}
We would like to thank the anonymous reviewers and the editor for all the helpful comments. We also 
thank Dr. Wenyi Zhang for helpful discussion on an earlier version of this paper.

\bibliographystyle{IEEEtran}
%\bibliography{LTE,WidebandScaling}

\begin{thebibliography}{99}

\bibitem{CDME15}
F.~G\'{o}mez-Cuba, J.~Du, M.~M\'{e}dard, and E.~Erkip, ``{Bandwidth Occupancy
  of Non-Coherent Wideband Fading Channels},'' in \emph{IEEE International
  Symposium on Information Theory (ISIT)}, 2015.

\bibitem{Pi2011}
Z.~Pi and F.~Khan, ``{An Introduction to Millimeter-Wave Mobile Broadband
  Systems},'' \emph{IEEE Communications Magazine}, vol.~49, pp. 101--107, Jun.
  2011.

\bibitem{PietBRPC:12}
P.~Pietraski, D.~Britz, A.~Roy, R.~Pragada, and G.~Charlton, ``{Millimeter Wave
  and Terahertz Communications: Feasibility and Challenges},'' \emph{ZTE
  Communications}, vol.~10, no.~4, pp. 3--12, 2012.

\bibitem{RanRapEr:14}
S.~Rangan, T.~T.~S. Rappaport, and E.~Erkip, ``{Millimeter-Wave Cellular
  Wireless Networks: Potentials and Challenges},'' \emph{Proceedings of the
  IEEE}, vol. 102, no.~3, pp. 366--385, Mar. 2014.

\bibitem{Rappaport2014-mmwbook}
T.~S. Rappaport, R.~W. {Heath Jr.}, R.~C. Daniels, and J.~N. Murdock,
  \emph{{Millimeter Wave Wireless Communications}}.\hskip 1em plus 0.5em minus
  0.4em\relax Pearson Education, 2014.

\bibitem{journals/tit/MedardG02}
M.~M\'{e}dard and R.~G. Gallager, ``{Bandwidth Scaling for Fading Multipath
  Channels},'' \emph{IEEE Transactions on Information Theory}, vol.~48, pp.
  840--852, Apr. 2002.

\bibitem{journals/twc/LozanoP12}
A.~Lozano and D.~Porrat, ``{Non-Peaky Signals in Wideband Fading Channels:
  Achievable Bit Rates and Optimal Bandwidth},'' \emph{IEEE Transactions on
  Wireless Communications}, vol.~11, pp. 246--257, Jan. 2012.

\bibitem{Zheng2007noncoherent}
L.~Zheng, D.~N.~C. Tse, and M.~Medard, ``{Channel Coherence in the Low-SNR
  Regime},'' \emph{IEEE Transactions on Information Theory}, vol.~53, pp.
  976--997, Mar. 2007.

\bibitem{Ray2007noncoherent}
S.~Ray, M.~Medard, and L.~Zheng, ``{On Noncoherent MIMO Channels in the
  Wideband Regime: Capacity and Reliability},'' \emph{IEEE Transactions on
  Information Theory}, vol.~53, pp. 1983--2009, Jun. 2007.

\bibitem{Hajek2002a}
V.~Subramanian and B.~Hajek, ``{Broad-band Fading Channels: Signal Burstiness
  and Capacity},'' \emph{IEEE Transactions on Information Theory}, vol.~48, pp.
  809--827, Apr. 2002.

\bibitem{Hajek2002b}
B.~Hajek and V.~Subramanian, ``{Capacity and Reliability Function for Small
  Peak Signal Constraints},'' \emph{IEEE Transactions on Information Theory},
  vol.~48, pp. 828--839, Apr. 2002.

\bibitem{journals/tit/TelatarT00}
I.~E. Telatar and D.~N. Tse, ``{Capacity and Mutual Information of Wideband
  Multipath Fading Channels},'' \emph{IEEE Transactions on Information Theory},
  vol.~46, pp. 1384--1400, Jul. 2000.

\bibitem{journals/tit/Verdu02}
S.~Verd\'{u}, ``{Spectral Efficiency in the Wideband Regime},'' \emph{IEEE
  Transactions on Information Theory}, vol.~48, pp. 1319--1343, Jun. 2002.

\bibitem{Medard2005}
C.~Luo, M.~Medard, and L.~Zheng, ``{On Approaching Wideband Capacity Using
  Multitone FSK},'' \emph{IEEE Journal on Selected Areas in Communications},
  vol.~23, no.~9, pp. 1830--1838, Sep. 2005.

\bibitem{Zhang07}
W.~Zhang and J.~N. Laneman, ``{How Good Is PSK for Peak-Limited Fading Channels
  in the Low-SNR Regime?}'' \emph{IEEE Transactions on Information Theory},
  vol.~53, pp. 236--251, Jan. 2007.

\bibitem{Wang09}
V.~Sethuraman, L.~Wang, B.~Hajek, and A.~Lapidoth, ``{Low-SNR Capacity of
  Noncoherent Fading Channels},'' \emph{IEEE Transactions on Information
  Theory}, vol.~55, pp. 1555--1574, Apr. 2009.

\bibitem{Durisi}
G.~Durisi, U.~Schuster, H.~Bolcskei, and S.~Shamai, ``{Noncoherent Capacity of
  Underspread Fading Channels},'' \emph{IEEE Transactions on Information
  Theory}, vol.~56, pp. 367--395, Jan. 2010.

\bibitem{SDBP09}
U.~Schuster, G.~Durisi, H.~Bolcskei, and H.~Poor, ``Capacity bounds for
  peak-constrained multiantenna wideband channels,'' \emph{IEEE Transactions on
  Communications}, vol.~57, pp. 2686--2696, Sep. 2009.

\bibitem{Hajek2005}
V.~Sethuraman and B.~Hajek, ``{Capacity Per Unit Energy of Fading Channels With
  a Peak Constraint},'' \emph{IEEE Transactions on Information Theory},
  vol.~51, pp. 3102--3120, Sep. 2005.

\bibitem{TSP04}
H.~Artes, G.~Matz, and F.~Hlawatsch, ``{Unbiased Scattering Function Estimators
  for Underspread Channels and Extension to Data-Driven Operation},''
  \emph{IEEE Transactions on Signal Processing}, vol.~52, no.~5, pp.
  1387--1402, May 2004.

\bibitem{Durisi12}
G.~Durisi, V.~Morgenshtern, and H.~Bolcskei, ``{On the Sensitivity of
  Continuous-Time Noncoherent Fading Channel Capacity},'' \emph{IEEE
  Transactions on Information Theory}, vol.~58, pp. 6372--6391, Oct. 2012.

\bibitem{Jindal2010}
N.~Jindal and A.~Lozano, ``{A Unified Treatment of Optimum Pilot Overhead in
  Multipath Fading Channels},'' \emph{IEEE Transactions on Communications},
  vol.~58, pp. 2939--2948, Oct. 2010.

\bibitem{Lapidoth05}
A.~Lapidoth, ``{On the asymptotic capacity of stationary Gaussian fading
  channels},'' \emph{IEEE Transactions on Information Theory}, vol.~51, pp.
  437--446, Feb. 2005.

\bibitem{Rappaport2013}
T.~S. Rappaport, S.~Sun, R.~Mayzus, H.~Zhao, Y.~Azar, K.~Wang, G.~N. Wong,
  J.~K. Schulz, M.~M. Samimi, and F.~Gutierrez, ``{Millimeter Wave Mobile
  Communications for 5G Cellular: It Will Work!}'' \emph{IEEE Access}, vol.~1,
  pp. 335--349, May 2013.

\bibitem{Ryu2015}
Z.~Zhang, J.~Ryu, S.~Subramanian, and A.~Sampath, ``{Coverage and Channel
  Characteristics of Millimeter Wave Band Using Ray Tracing},'' in \emph{IEEE
  International Conference on Communications (ICC)}, Jun. 2015, pp. 1380--1385.

\bibitem{mustafa2013mmWave}
M.~R. Akdeniz, Y.~Liu, M.~K. Samimi, S.~Sun, S.~Rangan, T.~S. Rappaport,
  E.~Erkip, and S.~Member, ``{Millimeter Wave Channel Modeling and Cellular
  Capacity Evaluation},'' \emph{IEEE Journal on Selected Areas in
  Communications}, vol.~32, pp. 1164--1179, Apr. 2013.

\end{thebibliography}

\end{document}